\documentclass[journal]{IEEEtran}

\usepackage[utf8]{inputenc}

\usepackage{cite}
\usepackage{multicol}
\usepackage[bookmarks=true]{hyperref}

\usepackage{mathrsfs} 
\usepackage{amsmath,amssymb,amsthm,amsfonts,graphicx,float}
\usepackage{mathtools}
\usepackage{color} 
\usepackage{tabularx}

\usepackage{algorithm}
\usepackage{algpseudocode}

\usepackage{tikz,pgfplots}
\pgfplotsset{compat=1.8}
\usepackage{pgfplotstable}
\usepgfplotslibrary{groupplots}
\usepackage{times}

\newtheorem{assumption}{Assumption}

\newtheorem{corollary}{Corollary}

\newtheorem{definition}{Definition}
\newtheorem{example}{Example}

\newtheorem{lemma}{Lemma}

\newtheorem{problem}{Problem}

\newtheorem{remark}{Remark}

\newcommand{\controlalgo}{\texttt{GTLProCo}}

\begin{document}

\title{ \LARGE Probabilistic Control of Heterogeneous Swarms Subject to Graph Temporal Logic Specifications: A Decentralized and Scalable Approach 
    \thanks{This material is based on work partly supported by the grants AFRL FA9550-19-1-0169 and DARPA D19AP00004.\newline 
    \indent F. Djeumou is with the Department of Electrical and Computer Engineering at the University of Texas at Austin, Austin, TX, USA. M. Cubuktepe and U. Topcu are with the Department of Aerospace Engineering  and Engineering Mechanics at the University of Texas at Austin, Austin, TX, USA. Email: \texttt{\{fdjeumou, mcubuktepe, utopcu\}@utexas.edu}.\newline
    \indent Z. Xu is with the School for Engineering of Matter, Transport and Energy, Arizona State University, Tempe, AZ, USA. Email: \texttt{xzhe1@asu.edu}.}}
\author{Franck Djeumou, Zhe Xu, Murat Cubuktepe, and Ufuk Topcu}
\maketitle
\author{Anonymous Authors}
\maketitle

\begin{abstract}
    We develop a probabilistic control algorithm, \controlalgo{}, for swarms of agents with heterogeneous dynamics and objectives, subject to high-level task specifications. 
    The resulting algorithm not only achieves decentralized control of the swarm but also significantly improves scalability over state-of-the-art existing algorithms. 
    Specifically, we study a setting in which the agents move along the nodes of a graph, and the high-level task specifications for the swarm are expressed in a recently-proposed language called graph temporal logic (GTL). 
    By constraining the distribution of the swarm over the nodes of the graph, GTL can specify a wide range of properties, including safety, progress, and response. 
    \controlalgo{}, agnostic to the number of agents comprising the swarm, controls the density distribution of the swarm in a decentralized and probabilistic manner.
    To this end, it synthesizes a time-varying Markov chain modeling the time evolution of the density distribution under the GTL constraints. 
    We first identify a subset of GTL, namely reach-avoid specifications, for which we can reduce the synthesis of such a Markov chain to either linear or semi-definite programs. 
    Then, in the general case, we formulate the synthesis of the Markov chain as a mixed-integer nonlinear program (MINLP). 
    We exploit the structure of the problem to provide an efficient sequential mixed-integer linear programming scheme with trust regions to solve the MINLP.  
    We empirically demonstrate that our sequential scheme is at least \emph{three} orders of magnitude faster than off-the-shelf MINLP solvers and illustrate the effectiveness of \controlalgo{} in several swarm scenarios.\looseness=-1
\end{abstract}

\section{Introduction} \label{sec:introduction}

Large numbers, or \emph{swarms}, of autonomous and heterogeneous agents can collaboratively achieve complex tasks that a single agent cannot. 
Such swarms have been used in the construction of a complex formation shape~\cite{kushleyev2013towards,rubenstein2014programmable}, opinion dynamics~\cite{kouvaros2016formal}, task allocations~\cite{berman2009optimized,liu2007towards}, surveillance, and search or rescue missions with ground or aerial vehicle swarms~\cite{cortes2004coverage,jaimes2008approach}. 
However, as the number of agents comprising a swarm increases, individual-agent-based control techniques for collective task completion become computationally intractable.  Besides, the heterogeneity of the agents makes the problem even harder due to the diverse dynamics and objectives.
Consequently, controlling a heterogeneous swarm to achieve some global task requirements remains a challenging problem.

 We propose an approach, which is agnostic to the number of agents comprising the swarm, to control a collective property of the swarm: its density distribution. 
 Specifically, we consider a setting in which the agents of a swarm move along the nodes of a graph~\cite{bondy1976graph}. 
 We sometimes refer to this graph as the \emph{configuration space}.     
 In this scenario, instead of controlling each agent individually, we propose an approach to control the time-varying density distribution of the swarm over the nodes of the graph. 
 Therefore, imposing constraints on the time evolution of the density distribution can express a set of collective behaviors for the swarm.

We specify the constraints on the evolution of the density distribution of the swarm using graph temporal logic (GTL)~\cite{xu2019graph}.
GTL, an extension of linear temporal logic (LTL)~\cite{pnueli1977temporal}, is an expressive language for high-level task specifications that focuses on the spatial-temporal properties of the node labels of a graph. 
Specifically, GTL can express  spatial-temporal properties on a graph more concisely than other logics such as alternating-time temporal logic (ATL)~\cite{Alur1997ATL} and LTL.
As an example, GTL can express properties such as ``whenever the density of the swarm in a node is less than $0.3$, eventually in the next $3$ time steps, at least two of its neighbor nodes have their density above $0.6$," by a formula with only a few propositions.
This property will result in a lengthier formula if expressed in either ATL or LTL.

We seek to synthesize controllers for swarms of heterogeneous agents subject to high-level task specifications expressed in GTL. We consider that the heterogeneity of the swarm is due to the diverse dynamics and objectives of its agents.
Then, the desired control algorithm should have the following properties: (a) correctness, i.e., the algorithm should enable the satisfaction of the GTL specifications; (b) scalability, i.e., the algorithm should scale with the size of the swarm and the size of the configuration space; and (c) distributed, i.e., the algorithm should return decentralized control laws to be executed by each agent comprising the swarm~\cite{accikmecse2012markov}.

We develop \controlalgo{} to control, in a probabilistic and decentralized manner, the time evolution of the density distribution of the swarm.
\controlalgo{} synthesizes a time-varying Markov chain~\cite{accikmecse2012markov,chattopadhyay2009supervised}, which models the time evolution of the density distribution in the configuration space. The obtained Markov chain is such that its states and transitions correspond to the nodes and edges of the configuration space.
On the agent level, the transition probability between two nodes is specified by the transition probability between the corresponding states of the synthesized Markov chain. 
Thus, the proposed formalism builds on the notion of transition probabilities between nodes of the configuration space, which is agnostic to the low-level individual dynamics or local interactions between agents as long as the transitions imposed by the synthesized Markov chain can be achieved.

\controlalgo{} computes Markov chains to control a swarm subject to GTL specifications through linear programs (LP), semi-definite programs (SDP), or mixed-integer linear programs (MILP). 
We first identify a subset of GTL, e.g., \emph{reach-avoid} specifications, for which we show that, depending on the graph structure, the synthesis of such a Markov chain can be formulated as either an LP or SDP problem. 
Then, in the general case, we formulate the synthesis problem as an NP-hard mixed-integer nonlinear programming (MINLP)~\cite{tawarmalani2002convexification} problem. 
Thus, we develop algorithms that can efficiently compute approximate solutions for the resulting MINLP. 
In the particular case where the agents move along the nodes of a complete graph~\cite{bondy1976graph}, we prove an equivalence between the feasibility of the MINLP and the feasibility of a mixed-integer linear program (MILP). 
For a non-complete graph, we adapt a \emph{sequential MILP} scheme that takes advantage of the structure of the problem to efficiently compute a locally optimal solution for the resulting MINLP. 
\controlalgo{} iteratively linearizes the nonconvex constraints around the solution of the previous iteration and leverages trust regions to account for the potential errors due to the linearization.

The complexity analysis of \controlalgo{} shows that, in the general case, its worst-case time complexity is polynomial in the size of the configuration space and only exponential in the size of the specifications.
The resulting complexity is a significant improvement over the existing approaches, which have exponential complexity in both the size of the configuration space and specifications.
We empirically demonstrate that even on relatively small problems, the sequential convex programming scheme is three orders of magnitude faster and more accurate than off-the-shelf MINLPs solvers~\cite{gurobi, belotti2009couenne,GleixnerEiflerGallyetal.2017}. 
Finally, we demonstrate the scalability and correctness of the developed algorithm in several gridworld scenarios involving heterogeneous swarms with thousands of agents. 

\controlalgo{} is scalable as it does not depend on the number of agents in the swarm. 
Furthermore, it is correct since, by construction, the resulting Markov chain enables the satisfaction of the GTL specifications.
Assuming that the transition time between two nodes is synchronized, each agent individually chooses the node to transit solely based on the transition probabilities of the synthesized Markov chain.
Henceforth, the Markov chain synthesized by the algorithm enables a decentralized control for each agent in the swarm.

\noindent {\textbf{Related work.}
Existing techniques for probabilistic density control~\cite{bandyopadhyay2017probabilistic, demir2015decentralized} of swarms based on the synthesis of Markov chains assumes homogeneity and do not consider complex behaviors, such as the ones induced by temporal logic specifications.
The approach in~\cite{demir2015decentralized} performs probabilistic decentralized control of swarms subject to ergodicity and upper-bound density constraints: A subset of GTL specifications. 
It provides an SDP formulation to find a Markov chain satisfying the constraints. 
In contrast, this paper improves the scalability of such approaches through an LP-based formulation for the Markov chain synthesis problem. 
Furthermore, to the best of our knowledge, this is the first paper to investigate probabilistic density control with temporal logic specifications.

The problem of synthesizing a controller for systems with multiple agents from a high-level temporal logic specification is considered in~\cite{kloetzer2008fully,kress2011correct,Allerton2019,swarmSTL2019,cubuktepe2020policy}. 
These papers define the specifications on the agent level and use an automata-based approach~\cite{kloetzer2008fully,wongpiromsarn2012receding} to compute a discrete controller satisfying the specifications over a finite abstraction of the system. 
However, it is expensive to compute such a finite abstraction, and the size of the automaton may be exponential in the length of the specifications while the synthesis of a controller can be double exponential in the length of the specifications.
Moreover, the length of the specifications depends on the size of the configuration space and may also grow exponentially with the number of controllable agents. 
Instead, this paper presents a synthesis algorithm with a worst-case time complexity that is only exponential in the size of the specifications. 

The synthesis of control algorithms for swarms subject to spatial and temporal logic specifications has also been considered in recent work~\cite{sahin2017provably,CensusSTL2016,haghighi2016robotic,haghighi2015spatel,moarref2019automated}. 
When considering spatial-temporal properties on a graph, GTL is more expressive than the spatial-temporal logics such as counting LTL~\cite{sahin2017provably} or SpaTeL~\cite{haghighi2015spatel,haghighi2016robotic}. 
Besides, the approaches based on these logics are significantly less scalable than the proposed approach, and most of them require a central unit to assign targets to individual agents.
Specifically, the number of integer variables in the optimization problems resulting from counting LTL-based, GR(1)-based~\cite{moarref2019automated}, and SpaTeL-based approaches depends on the size of the specifications. 
Besides, it exhibits quadratic dependency on the size of the considered abstraction. 
In contrast, the number of integer variables in the proposed approach depends only on the size of the specifications.
}

We make several extensions over our conference paper~\cite{Djeumou-RSS-20}. 
First, \controlalgo{} enables probabilistic control of heterogeneous swarms subject to infinite-horizon GTL formulas, while the conference paper assumes homogeneous swarms and finite-horizon GTL formulas. 
Second, this paper identifies a subset of GTL for which LP and SDP formulations are sufficient to solve the control problem. 
Finally, the conference paper relies on a coordinate descent algorithm to solve the resulting MINLP, which convergence to an accurate and feasible solution is highly dependent on the proximity of the starting point of the algorithm to an optimal solution. 
In contrast, the sequential scheme developed in this paper is faster, more accurate, and robust to the choice of the starting point. 


\noindent {\textbf{Contributions.}  We make the following contributions: (a) we present a novel, correct-by-construction, scalable, and decentralized algorithm for controlling swarms of autonomous and heterogeneous agents subject to GTL specifications; (b) we develop an algorithm, \controlalgo{}, based on LP, SDP, and MILP formulations to efficiently tackle the control problem, and we provide a worst-case time complexity analysis of \controlalgo{}; (c) we evaluate the developed algorithm on numerical examples involving a large number of agents.
}

\section{Preliminaries} \label{sec:prel}

\noindent {\textbf{Notation.} $\boldsymbol{0}$ is the zero matrix or vector of appropriate dimensions. $\boldsymbol{1}$ denotes a vector with all elements equal to $1$ of appropriate dimensions. $e_i$ is a vector of appropriate dimensions with its i-th entry $1$ and its other entries $0$. $A^{\mathrm{T}}$ denotes the transpose of a matrix $A$. $A_{i,j} = A[i,j] = e_i^{\mathrm{T}} A e_j$ for a matrix $A$. $x_i = x[i] = e_i^{\mathrm{T}} x$ for a vector x.  Comparisons (e.g., $\ge$) between matrices or vectors are conducted element-wise. The operator $\odot$ represents the element-wise product, and $[x_1;x_2]$ is the vector obtained by stacking vectors $x_1$ and $x_2$.
}

\subsection{Markov Chain-based Control of Homogeneous Swarms}
We present the definitions and assumptions used in the Markov chain approach to control swarms of autonomous and \emph{homogeneous} agents. Note that most of the definitions in this section can be found in the existing literature~\cite{accikmecse2012markov,bandyopadhyay2017probabilistic}.

\begin{definition}[\textsc{Bins}]
    The configuration space over which the state of an agent is distributed is denoted as $\mathcal{R}$. It is assumed that $\mathcal{R}$ is partitioned into $n_{\mathrm{r}}$ disjoint subspaces called \textit{bins}.
    \begin{equation*}
        \mathcal{R} = \cup_{i=1}^{n_{\mathrm{r}}} \mathcal{R}_i, \text{ s.t. } \mathcal{R}_i \cap \mathcal{R}_j = \emptyset, \: \text{ for all } i \neq j.
    \end{equation*}
    Each bin $\mathcal{R}_i$ (also referred to as bin $i$) represents a predefined range of the state of an agent, e.g., position, behavior, etc.
    \label{def:bins_subspaces}
\end{definition}

\begin{definition}[\textsc{State of an agent}]
    We denote by $N_{\mathrm{a}}$ the number of agents in the swarm. We define $r^m (t) \in \{0,1\}^{n_{\mathrm{r}}}$ as the state of agent $m$ at time $t$. If $r^m (t)$ belongs to the bin $\mathcal{R}_i$, for some $i \in \{1,\hdots,n_{\mathrm{r}}\}$, then $r^m(t) = e_i$.
    \label{def:agent_state}
\end{definition}

\begin{definition}[\textsc{Motion constraints}]
     The state of each agent can transition, between two consecutive time steps, from a bin to only certain bins because of the dynamics or the environment. These motion constraints are specified by the fixed matrix $A_{\mathrm{adj}} \in \{0,1\}^{n_{\mathrm{r}} \times n_{\mathrm{r}}}$, called an \textit{adjacency matrix}. Each component of $A_{\mathrm{adj}}$ is given by
    \begin{equation*}
        A_{\mathrm{adj}}[i,j] = 
            \begin{dcases}
            1 \quad \text{if the transition from bin $\mathcal{R}_i$} \\
            \  \quad \text{to bin $\mathcal{R}_j$ is allowed}, \\
            0 \quad \text{if this transition is not allowed.}\end{dcases}
    \end{equation*}
    Equivalently, the topology of the bins can be modeled as a graph $G = (V,E)$ where $V = \{v_1,\hdots,v_{n_{\mathrm{r}}}\}$ is the set of bins, and $E \subseteq V \times V$ is the set of edges such that  $(v_i,v_j) \in E$ if and only if $A_{\mathrm{adj}}[i,j] = 1$, $\forall i,j \in \{1,\hdots,n_{\mathrm{r}}\}$.
    \label{def:motion_constraint}
\end{definition}

In the rest of the paper, when we refer to an agent belonging to a bin, we mean that its state belongs to that bin. Similarly, when we refer to an agent transiting between bins, we suggest that its state transits between these bins.

\begin{example}
    Consider a swarm scenario where the state of an agent is its position, and the physical configuration space is partitioned into $n_{\mathrm{r}} = 3$ bins. Consider that $A_{\mathrm{adj}} = [[1,1,0]^\mathrm{T};[1,1,1]^\mathrm{T};[0,1,1]^\mathrm{T}]$. Having that $A_{\mathrm{adj}}[1,1]= A_{\mathrm{adj}}[1,2]=1$, and $A_{\mathrm{adj}}[1,3] = 0$ enforces agents in bin $\mathcal{R}_1$ to either stay in $\mathcal{R}_1$ or transit to $\mathcal{R}_2$ between two consecutive time steps. The corresponding graph $G = (V,E)$ is given by $V = \{v_1,v_2,v_3\}$, where the nodes $v_1,v_2$, and $v_3$ represent respectively the bins $\mathcal{R}_1,\mathcal{R}_2$, and $\mathcal{R}_3$. The set of edges is given by $E = \{(v_1,v_1),(v_1,v_2),(v_2,v_1),(v_2,v_2),$ $(v_2,v_3),(v_3,v_2),(v_3,v_3)\}$.
    \label{running-example}
\end{example}

\begin{definition}[\textsc{Density distribution of the swarm}]
    The density distribution $x(t) \in \mathbb{R}^{n_{\mathrm{r}}}$ of a swarm is a column-stochastic vector, i.e. $x(t) \geq \boldsymbol{0} \text{ and } \boldsymbol{1}^{\mathrm{T}} x(t) = 1$, such that a component $x_i(t)$ is the proportion of agents in bin $\mathcal{R}_i$ at time $t$:
    \begin{equation*}
        x_i(t) := \frac{1}{N_{\mathrm{a}}} \sum_{m=1}^{N_{\mathrm{a}}} r^m_i(t).
    \end{equation*}
    \label{def:swarm_density}
\end{definition}

\begin{definition}[\textsc{Transition policy of an agent}]
    At time $t$, the agent $m$ transits from bin $\mathcal{R}_j$ to bin $\mathcal{R}_i$ with probability
    \begin{equation*}
        M^m_{i,j}(t) = Pr(r^m_i(t+1) = 1 | r^m_j(t) = 1), 
    \end{equation*}
    where $M^m (t) \in \mathbb{R}^{n_{\mathrm{r}} \times n_{\mathrm{r}}}$ is a column-stochastic matrix, i.e. $\boldsymbol{1}^{\mathrm{T}} M^m (t) = \boldsymbol{1}^{\mathrm{T}}, M^m(t) \geq \boldsymbol{0}$. We refer to $M^m(t)$ as the time-varying Markov matrix of agent m at time $t$.
    \label{def:mc_dynamic}
\end{definition}

\begin{remark}
    Under the motion constraints given by $A_{\mathrm{adj}}$, the transitions between some bins may not be allowed. For agent m, $M^m_{i,j}(t)$ is the probability of transition from bin $\mathcal{R}_j$ to bin $\mathcal{R}_i$. Hence, $M^m_{i,j}(t) = 0$ if $A_{\mathrm{adj}} [j,i] = 0$.
    \label{remark-motion-constraints}
\end{remark}

In Example~\ref{running-example}, if $M^m(t)$ is the time-varying Markov matrix of agent $m$ at time $t$, then $M^m_{2,1}(t)$ gives the probability of agent $m$ to transit from bin $\mathcal{R}_1$ to bin $\mathcal{R}_2$ in one time step. Moreover, having $A_{\mathrm{adj}}[1,3] = 0$ enforces that $M^m_{3,1}(t)=0$.

In this section, we focus on methods that ensure that each agent of the homogeneous swarm has the same time-varying Markov matrix at any given time $t$, i.e., $M^1(t) = \cdots = M^{N_{\mathrm{a}}}(t)= M(t)$. When the agents independently choose their transitions between bins using $M(t)$, two mathematical interpretations are given for $x(t)$~\cite{accikmecse2012markov}: (a) $x(t)$ is the vector of expected ratio of the number of agents in each bin; (b) the ensemble of agent state, $\{r^k(t)\}_{k=1}^{N_{\mathrm{a}}}$, has a distribution that approaches $x(t)$ with probability one as $N_{\mathrm{a}}$ increases towards infinity (due to the law of large numbers). As a consequence, the dynamics of the density distribution of the swarm can be modeled by~\cite{accikmecse2012markov,demir2015decentralized}
\begin{equation}
    x(t+1) = M(t) x(t),
    \label{eq:swarm_dyn}
\end{equation}
as $N_{\mathrm{a}}$ increases towards infinity. The Markov chain approach for the control of swarms relies on the synthesis of a time-varying Markov matrix $M(t)$ such that the time evolution of the density distribution of the swarm is given by~\eqref{eq:swarm_dyn}.

\subsection{Graph Temporal Logic} \label{sec:gtl}

Let $G=(V, E)$ be a graph, where $V$ is a finite set of nodes and $E$ is a finite set of edges. We use $\mathcal{X}$ to denote a (possibly infinite) set of node labels. $\mathbb{T}=\{0, 1, \dots\}$ is a discrete set of time indices. A graph with node labels is also called a \textit{labeled graph}. A trajectory $g : V\times\mathbb{T}\rightarrow \mathcal{X}$ on the graph $G$ denotes the time evolution of the node labels.

In the swarm scenario, we focus on labelled versions of $G$ where each of its nodes is labelled with a given function of the density distribution of the swarm. That is, the graph trajectory $g$ at node $v_i$ and time $t$ is given by $g(v_i,t) = f_i(x(t))$, where $f_i : \mathbb{R}^{n_\mathrm{r}} \mapsto \mathcal{X}$ is known and specific to $v_i$. For example, for the remainder of this section, consider the following labelling on $G$ of Example~\ref{running-example}: $f_1(x) = [x_1, x_1 - x_3]^\mathrm{T}$, $f_2(x) = [x_2, x_3-x_1-x_2]^\mathrm{T}$, and $f_3(x) = [x_3, x_3-x_2]^\mathrm{T}$.

An \textit{atomic node proposition} is a predicate on $\mathcal{X}$, i.e. a Boolean valued map from $\mathcal{X}$. We use $\pi$ to denote an atomic node proposition, and $\mathcal{O}(\pi)$ to denote the subset of $\mathcal{X}$ for which $\pi$ is true.

We define that a graph trajectory $g$ satisfies the atomic node proposition $\pi$ at a node $v$ at time index $k$, denoted as $(g,v,k)\models\pi$, if and only if $g(v,k) \in \mathcal{O}(\pi)$. In Example~\ref{running-example}, using the labelling described above, if $x(0) = [0.3,0.3,0.4]^{\mathrm{T}}$ and $\pi = (y \leq [0.3, 0]^\mathrm{T})$ with $y$ a symbolic representation of $f_i$, then $\pi$ is satisfied by $g$ at time index $0$ at nodes $v_1$ and $v_2$.

\begin{definition}[\textsc{Neighbor operator}]
	\label{def:neighbor}
	Given a graph $G$, the \textit{neighbor operation} $\bigcirc$ : $2^{V}\rightarrow 2^{V}$ is defined as
	\[
	\begin{split}
	   \bigcirc(V') = \{v  \in V| \exists v'\in V'  \text{ s.t. }(v',v)\in E\}.
	 \end{split} \]
	Intuitively, $\bigcirc(V')$ consists of nodes that can be reached from $V'$. Note that neighbor operations can be applied successively. In Example~\ref{running-example}, we have $\bigcirc(\{v_1\}) = \{v_1,v_2\}$.
\end{definition}

We refer to a graph trajectory as a trajectory $g : V \times \{0,\hdots,T_{\mathrm{f}}\} \rightarrow \mathcal{X}$, where $T_{\mathrm{f}} \in \mathbb{T} \cup \{+\infty\}$. Graph trajectories are sufficient to satisfy (resp. violate) GTL formulas. We define the syntax of a  GTL formula $\varphi$ recursively as    
\begin{center}
	$\varphi:=\pi~|~\neg\varphi_1~|~X\varphi_1~|~\varphi_1\wedge\varphi_2~|~\varphi_1\mathcal{U}\varphi_2~|~\exists^{N}(\bigcirc \cdots \bigcirc)\varphi_1$, 
\end{center}
where $\pi$ is an atomic         
node proposition, $\exists^{N}(\bigcirc \cdots \bigcirc)\varphi$ reads as \textquotedblleft there exist at least $N$ nodes under the neighbor operation $\bigcirc \cdots \bigcirc $ that satisfy $\varphi$ \textquotedblright, $\lnot$ and $\wedge$ stand for negation and conjunction respectively, $X$ is the temporal operator \textquotedblleft next\textquotedblright, and $\mathcal{U}_{\le i}$ is the temporal operator \textquotedblleft until". We can also derive $\vee$ (disjunction), $\Rightarrow$ (implication), $\Diamond$ (eventually),  $\Box$ (always), $\Box \Diamond$ (always enventually), and $\Diamond \Box$ (eventually always) from the above-mentioned operators~\cite{Baier2008}, e.g. 
\begin{align*}
    \Diamond \varphi = \mathrm{True} \: \mathcal{U} \: \varphi,  \: \: &
    \Box \varphi = \neg \Diamond \neg \varphi.
\end{align*}

The satisfaction relation $(g, v, t)\models\varphi$ for a graph trajectory $g$ at node $v$ at time index $t$ with respect to a GTL formula $\varphi$ is defined recursively by
\begin{equation*}
    \begin{aligned}
        &(g, v, t)\models\pi \: &\text{iff } & \: g(v, t)\in\mathcal{O}(\pi),\\
        &(g, v, t)\models\lnot\varphi \: &\text{iff }& \: (g, v, t)\not\models\varphi,\\
        &(g, v, t)\models X\varphi \: &\text{iff } & \:(g, v, t+1)\models\varphi,\\
        &(g, v, t)\models\varphi_{1}\wedge\varphi_{2} \: &\text{iff }& \: (g, v, k)\models\varphi_{1} \: \text{and}~(g, v, t)\models\varphi_{2},\\
        &(g, v, t)\models\varphi_{1}\mathcal{U}\varphi_{2} \: &\text{iff }& \exists t'\geq t, \mbox{s.t.}~(g, v, t')\models\varphi_{2} \text{ and}\\
        &          &  &(g, v, t^{\prime\prime})\models\varphi_{1}, \forall t \leq t^{\prime\prime} < t',
    \end{aligned}
\end{equation*}
\begin{equation*}
    \begin{aligned}
        (g, &v, t) \models\exists^{N}(\bigcirc \cdots \bigcirc)\varphi \: \mbox{iff}  \: \: \exists v_1, \dots, v_N~(v_i\neq v_j ~\mbox{for} \\
        &i\neq j),~\mbox{s.t.}, \forall i, v_i\in \bigcirc \cdots \bigcirc (\{v\}),~\mbox{and}~ (g, v_i, t)\models\varphi.             
    \end{aligned}                                             
\end{equation*}

Intuitively, a graph trajectory $g$ satisfies $\exists^{N}(\bigcirc \cdots \bigcirc)  \varphi$ at a node $v\in V$ at time index $k$, if there exist at least $N$ nodes in $\bigcirc \cdots \bigcirc(\{v\})$ where $\varphi$ is satisfied by $g$ at time index $k$. Note that, by definition, if $\bigcirc \cdots \bigcirc (\{v\})$ consists of fewer than $N$ nodes, then $\exists^{N}(\bigcirc \cdots \bigcirc )\varphi$ is false. In Example~\ref{running-example}, if $x(0) = [0.3,0.3,0.4]^{\mathrm{T}}$, then the nodes that satisfy $\exists^{1} \bigcirc (y \geq [0.2, 0]^\mathrm{T})$ at time index $0$ are $v_2$ and $v_3$.

We also define that a graph trajectory $g$ satisfies $\varphi$ at node $v$, denoted as $(g, v)\models\varphi$, if $g$ satisfies $\varphi$ at node $v$ at time $0$.



\section{Problem Formulation} \label{sec:prob-formulation}
In this section, we first specify the link between a graph trajectory satisfying a graph temporal logic (GTL) formula and the time evolution of the density distribution of a swarm. Then, we formulate the problem of controlling the density distribution of a swarm subject to GTL, as the problem of synthesizing time-varying Markov matrices.

In the remainder of the paper, we assume a configuration space divided into $n_\mathrm{r}$ bins, and we consider that the swarm of heterogeneous agents can be partitioned into $m$ smaller swarms of homogeneous agents. Typically, such partitioning enables to regroup agents in the swarm that might have the same dynamics, objectives, or motion constraints.
\begin{definition}[\textsc{Sub-swarms}]\label{def:sub-swarm}
    Given $s \in \{1, \hdots, m\}$, the $s$-th sub-swarm is a collection of homogeneous agents with motion constraints given by the adjacency matrix $A_{\mathrm{adj}}^s \in \{0,1\}^{n_\mathrm{r} \times n_\mathrm{r}}$, its density distribution denoted by $x^s(t)$, and the graph induced by $A_{\mathrm{adj}}^s$ (Definition~\ref{def:motion_constraint}) denoted by $G^s=(V,E^s)$.
\end{definition}    
Thus, we define GTL specifications over the heterogeneous swarm as joint constraints on the time evolution of the density distributions of the sub-swarms.

\begin{definition}[\textsc{GTL specifications}]
    Let $G=(V,E)$ with $E = \cup_{s=1}^m E^s$ be the graph obtained by considering the motion constraints of all the sub-swarms. By labeling each node $v_i \in V$ with a function $f_i : (\mathbb{R}^{n_\mathrm{r}})^m \mapsto \mathcal{X}$ of all the density distribution $x^s(t)$ of the sub-swarms, we define GTL specifications on the swarm as GTL formulas on the obtained labelled graph $G$.
    \label{def:gtl-swarm}
\end{definition}

Definition~\ref{def:gtl-swarm} specifies that a graph trajectory $g$ on the labeled graph $G$, at node $v_i$ and time index $t \in \mathbb{T}$, is given by $g(v_i,t) = f_i(x^1(t),\hdots,x^m(t))$.

\begin{assumption}
    For all $i \in \{1,\hdots,m\}$, the function $f_i$ associated to the node label of $v_i$ is an affine function, and for every atomic node proposition $\pi$ of a given GTL formula, $\mathcal{O}(\pi) \subseteq \mathcal{X}$ is a convex polyhedra.
    \label{ass:atomic-prop-halfspace}
\end{assumption}

\begin{problem}
    Given the adjacency matrices $A_{\mathrm{adj}}^s$ for all $s\in \{1,\hdots,m\}$, the induced labelled graph $G=(V,E)$ (Definition~\ref{def:gtl-swarm}), the initial density distributions $x^1(0), \hdots, x^m(0)$ for all sub-swarms, a set $V' \subseteq V$, and a GTL formula $\varphi$ on $G$, compute the time-varying Markov matrices $M^s(t)$ for all $s \in \{ 1,\hdots,m\} $  such that the followings are true:
    \begin{enumerate}
        \item The motion constraints are satisfied by all sub-swarms.
        \item $(g , v_i) \models \varphi$ for all $v_i \in V'$, where $g$ is induced by the combined evolution of $M^s(t)$ for all $s \in \{1,\hdots,m\}$.
        \item A linear cost function $\mathcal{C} : (\mathbb{R}^{n_\mathrm{r}})^m \times (\mathbb{R}^{n_\mathrm{r} \times n_\mathrm{r}})^m \mapsto \mathbb{R}$ is minimzed over time.
    \end{enumerate}
    \label{prob:mc-approach-gtl}
\end{problem}

\begin{remark}
    According to~\eqref{eq:swarm_dyn}, $M^s(t)$ dictates the evolution of the density distribution $x^s(t)$ of the $s$-th sub-swarm. Thus, the matrices $M^s(t)$ for all $s\in\{1,\hdots,m\}$ also specify the graph trajectory  of $G$ since by Definition~\ref{def:gtl-swarm} the trajectory at $v_i$ is given by $f_i$, a function of all $x^s(t)$.
\end{remark}

The cost function $\mathcal{C}$ enables to distinguish among the trajectories satisfying $\varphi$. Typically, we seek to minimize over a time horizon $T$, $\sum_{t=0}^{T} \mathcal{C}((x^s(t))^m_{s=1},(M^s(t))^m_{s=1})$.

As a toy example, consider Example~\ref{running-example} with two sub-swarms having the same motion constraints $A_{\mathrm{adj}}^1=A_{\mathrm{adj}}^2$. The initial density distributions are $x^1(0) = [0.3,0.3,0.4]^T$ and $x^2(0) = [0.3,0.4,0.3]^T$. We label nodes $v_1$, $v_2$, and $v_3$ with the functions $f_1(x^1,x^2) = x^1_1+x^2_1$, $f_2(x^1,x^2) = x^2_2 - 2x_2^1$, and $f_3(x^1,x^2) = x^1_3$. We consider the GTL formula $\varphi_1 = X (\Box (y=0))$ specified for node $v_1$ and $v_2$ (bin $\mathcal{R}_1$ and $\mathcal{R}_2$), and no cost function. Recall that $y$ is a symbolic representation of $f_i$. That is, $\varphi_1$ specified at node $v_i$ can also be written as $\varphi_1 = X (\Box (f_i(x^1,x^2)=0))$. Intuitively, $\varphi_1$ specified for $v_1$ and $v_2$ means that starting from time index $1$, there should always be no agents from both sub-swarms in bin $\mathcal{R}_1$. Further, the density of the $2$nd sub-swarm in bin $\mathcal{R}_2$ is always twice the density of the $1$st sub-swarm in bin $\mathcal{R}_2$. Markov matrices $M^1(t)$ and $M^2(t)$ solution to Problem~\ref{prob:mc-approach-gtl} are given by
\begin{align*}
    M^1(0) = \begin{bmatrix} 0 & 0 & 0 \\ 1 & 0 & 0 \\ 0 & 1 & 1 \end{bmatrix},
    M^2(0) = \begin{bmatrix} 0 & 0 & 0 \\ 1 & 0.75 & 0 \\ 0 & 0.25 & 1 \end{bmatrix}, M^s(t) = \mathbb{I}_{3},
\end{align*}
where $s\in\{1,2\}$, $t \geq 1$ and $\mathbb{I}_{3}$ is the identity matrix of dimension $3$. Intuitively, at time $0$, the agents of each sub-swarm in $\mathcal{R}_1$ must move to $\mathcal{R}_2$ with probability $1$, the agents of the $1$st sub-swarm in $\mathcal{R}_2$ must move to $\mathcal{R}_3$ with probability $1$ while the agents of the $2$nd sub-swarm in $\mathcal{R}_2$ moves to $\mathcal{R}_3$ with probability $0.75$ and remains in $\mathcal{R}_2$ with probability $0.25$. For all $t \geq 1$, the agents in each sub-swarms remains in their current bin. The reader can check that with $x^s(t+1) = M^s(t) x^s(t)$ for $t \geq 0$ and $s \in \{1,2\}$, we have $x^1_1(t) =x^2_1(t) = 0$  and $x^2_2(t) = 2 x^1_2(t)$ hold for all $t \geq 1$. Thus, $\varphi_1$ is satisfied at nodes $v_1$ and $v_2$ at time index $0$.

\section{MINLP Formulation} \label{sec:minlp-formulation}

In this section, we do not make any assumptions on the GTL specifications or the structure of the graph, and we formulate Problem~\ref{prob:mc-approach-gtl} as a mixed-integer nonlinear programming (MINLP) problem containing $M^s(t)$ as the variables.

\subsection{Stochasticity and Motion Constraints}

The desired time-varying Markov matrices $M^s(t)$ at time index $t$ and for all $s$ are column-stochastic matrices, i.e.,
\begin{equation}\label{eq:stochast-constr}
    \boldsymbol{1}^{\mathrm{T}} M^s(t) = \boldsymbol{1}^{\mathrm{T}}.
\end{equation}

From Remark~\ref{remark-motion-constraints}, we have that $M^s_{i,j}(t) = 0$ if $A^s_{\mathrm{adj}} [j,i] = 0$, and $M^s_{i,j}(t) \geq 0$ otherwise. Thus, for all $s \in \{1,\hdots,m\}$,
\begin{align}
    (\boldsymbol{1}\boldsymbol{1}^{\mathrm{T}} - (A^s_{\mathrm{adj}})^{\mathrm{T}}) \odot M^s(t) &= \boldsymbol{0}, \label{eq:zero-transition}\\
    M^s(t) &\geq \boldsymbol{0}. \label{eq:pos-trans}
\end{align}

\subsection{Mixed-Integer Encoding of GTL Formulas}
In this section, we ignore the constraints implied by the dynamics~\eqref{eq:swarm_dyn}, and we build on the work in~\cite{wolff2014optimization} to provide a mixed-integer linear program (MILP) for finding graph trajectories satisfying a GTL formula $\varphi$ on the labelled graph $G = (V,E)$, where $V = \{v_1,\hdots,v_{n_\mathrm{r}}\}$. 

Although satisfying an infinite-horizon GTL formulas requires a graph trajectory of infinite length, we design \emph{periodic} trajectories to capture the infinite length requirement.
\begin{definition}[\textsc{$(k_\mathrm{p},l_\mathrm{p})$-periodic graph trajectory}]\label{def:periodic-trajectory}
    A graph trajectory $g_\mathrm{p}$ is $(k_\mathrm{p},l_\mathrm{p})$-periodic if
    \begin{align}
        g_\mathrm{p}(v_i, k_\mathrm{p}) = g_\mathrm{p}(v_i, l_\mathrm{p}-1), \label{eq:loop-cond}
    \end{align}
    and for all $t \in \{l_\mathrm{p},\hdots,k_\mathrm{p}-1\}$, we have that
    \begin{align}
        g_\mathrm{p}(v_i, t + (k_\mathrm{p}-l_\mathrm{p}+1)q) = g(v_i, t), \: \forall q \geq 0,  \label{eq:periodicity-cond}
    \end{align}
    where $l_\mathrm{p}, k_\mathrm{p} \in \mathbb{T}$ are such that $0<l_\mathrm{p} \leq k_\mathrm{p}$, and $v_i \in V$ is a node of $G$. Thus, such trajectory can be seen as a finite sequence of length $k_\mathrm{p}$, where a loop is introduced between the $(k_\mathrm{p}-1)$-th and the $(l_\mathrm{p}-1)$-th elements of the sequence.
\end{definition}

As a consequence, given a node $v_i \in V$ and a length $k_\mathrm{p} \in \mathbb{T}$, we seek for mixed-integer linear constraints that are satisfiable if and only if there exists $l_\mathrm{p}\in (0,k_\mathrm{p}]$ and a  $(k_\mathrm{p},l_\mathrm{p})$-periodic graph trajectory $g_\mathrm{p}$ such that $(g_\mathrm{p}, v_i) \models \varphi$. Specifically, given $t \in \{0,\hdots,k_\mathrm{p}\}$ and the formula $\varphi$, we construct the equivalent mixed-integer constraints that encode $(g_\mathrm{p},v_i,t) \models \varphi$ by induction on $t$ and $\varphi$ as follows.

Before going through the induction, we first encode the loop constraint resulting from the periodicity of $g_\mathrm{p}$. To this end, we introduce $k_\mathrm{p}$ binary variables $l_1,\hdots,l_{k_\mathrm{p}}$ which determine where the graph trajectory loops. The variables are such that there is a unique $l_j$ satisfying $l_j=1$ and such $l_j$ enforces $g_\mathrm{p}(v_i, k_\mathrm{p}) = g_\mathrm{p}(v_i, j-1)$. Thus, since $g_\mathrm{p}$ is affine in $x^s$, such loop constraint can be encoded as the mixed-integer constraints
\begin{align}
    &l_1 + l_2 + \cdots + l_{k_{\mathrm{p}}} = 1, \label{eq:unique-lj}\\
    &g_\mathrm{p}(v_i, k_\mathrm{p}) \leq g_\mathrm{p}(v_i, j-1) + P (1 - l_j), \: j = 1,\hdots,k_\mathrm{p}, \label{eq:loop1}\\
    &g_\mathrm{p}(v_i, k_\mathrm{p}) \geq g_\mathrm{p}(v_i, j-1) - P (1 - l_j), \: j = 1,\hdots,k_\mathrm{p}, \label{eq:loop2}
\end{align}
where $P>0$ is a sufficiently large positive number, and recall that $g(v_i,t) = f_i(x^1(t),\hdots,x^m(t))$.

In the case $\varphi = \pi$, where $\pi$ is an atomic node proposition, $\mathcal{O}(\pi) \subseteq \mathcal{X}$ is a convex polyhedra by Assumption~\ref{ass:atomic-prop-halfspace}. Thus, using the halfspace representation, we can write $\mathcal{O}(\pi)$ as the intersection of finite number of halfspaces. That is, there exists a matrix $A$ and vector $b$ of appropriate dimensions such that $g_\mathrm{p}(v_i,t) \in \mathcal{O}(\pi)$ if and only if $A g_\mathrm{p}(v_i,t) \leq b$. That is, $g_\mathrm{p}(v_i,t) \in \mathcal{O}(\pi)$ if and only if $ A f_i(x^1(t),\hdots,x^m(t)) \leq b$. Since $f_i$ is a linear function of $x^1, \hdots, x^m$,  the last inequality is also a linear inequality. Thus, if the binary variable $\varphi^t \in \{0,1\}^d$ encodes the result of the query $(g_\mathrm{p}, v_i, t) \models \pi$, the equivalent mixed-integer constraints are given by~\eqref{eq:unique-lj}--\eqref{eq:loop2},
\begin{align}
    A f_i(x^1(t),\hdots,x^m(t)) &\leq b + P (1-\varphi^t), \\
    A f_i(x^1(t),\hdots,x^m(t)) &> b - P \varphi^t.
\end{align} 

For $\varphi = \lnot \varphi_1$, let $\varphi^t_1$ (binary or continuous) encodes the result of  the query $(g_\mathrm{p}, v_i, t) \models \varphi_1$. That is, $\varphi^t_1 = 1$ if and only if $(g_\mathrm{p}, v_i, t) \models \varphi_1$. By induction hypothesis on $\varphi_1$, there exists an equivalent mixed-integer constraint denoted by $[[A^{1},b^{1},w^p,\varphi^t_1]]$ such that $A^{1} [x^1(t_{1});\hdots; x^m(t_{1}); w^p;\varphi^t_1] \leq b^{1}$, for some $A^{1}$ and $b^{1}$ of appropriate dimensions. $w^p \in [0,1]^p$ combines binary and continuous variables, where $p \in \mathbb{N}$. Thus, if the continuous variable $\varphi^t \in [0,1]$ encodes the result of the query $(g_\mathrm{p}, v_i, t) \models \lnot \varphi_1$,  using the definition of the $\lnot$ operator, the equivalent mixed-integer constraints are given by~\eqref{eq:unique-lj}--\eqref{eq:loop2},
\begin{align}
    &A^{1} [x^1(t_{1});\hdots; x^m(t_{1}); w^p; \varphi^t_1] \leq b^{1}, \\
    &\varphi^t =  1 - \varphi^t_1
\end{align}

In the case $\varphi = \varphi_1\wedge\varphi_2$, let $\varphi^t_1$ and $\varphi^t_2$, both either binary or continuous variables, encode the result of the queries $(g_\mathrm{p}, v_i, t) \models \varphi_1$ and $(g_\mathrm{p}, v_i, t) \models \varphi_2$, respectively. By induction hypothesis, there exists equivalent mixed-integer constraints $[[A^{1},b^{1},w^p,\varphi^t_1]]$ and $[[A^{2},b^{2},z^q,\varphi^t_2]]$ for satisfiablity of $\varphi_1$ and $\varphi_2$. Thus, if the continuous variable $\varphi^t \in [0,1]$ encodes the result of the query $(g_\mathrm{p}, v_i, t) \models \varphi_1 \wedge \varphi_2$,  using the definition of the $\wedge$ operator, the equivalent mixed-integer constraints are given by~\eqref{eq:unique-lj}--\eqref{eq:loop2},
\begin{align}
    &A^{1} [x^1(t_{1});\hdots; x^m(t_{1}); w^p;\varphi^t_1] \leq b^{1}, \\
    &A^{2} [x^1(t_{1});\hdots; x^m(t_{1}); z^q;\varphi^t_2] \leq b^{2}, \\
    &\varphi^t \leq \varphi^t_j, \: j=1,2, \\
    &\varphi^t \geq \varphi^t_1 + \varphi^t_2 - 1.
\end{align}

In the case $\varphi = X \varphi_1$, let $[[A^{1},b^{1},w^p,\varphi^t_1]]$ be the equivalent mixed-integer constraint obtained by induction on $\varphi_1$, where $\varphi^t_1$ (binary or continuous) encodes the result of the query $(g_\mathrm{p}, v_i, t) \models \varphi_1$. If $t < k_\mathrm{p}$ and the binary variable $\varphi^t$ encodes the result of the query $(g_\mathrm{p}, v_i, t) \models X \varphi_1$, then the equivalent mixed-integer constraint is given by~\eqref{eq:unique-lj}--\eqref{eq:loop2},
\begin{align}
    &A^{1} [x^1(t_{1});\hdots; x^m(t_{1}); w^p; \varphi^t_1] \leq b^{1} + P (1- \varphi^t), \nonumber \\
    &A^{1} [x^1(t_{1});\hdots; x^m(t_{1});w^p; \varphi^t_1] > b^{1} - P \varphi^t, \nonumber
\end{align}
where $t_1 = t+1$ and $P >0$ is a sufficiently large number. However, if $t = k_\mathrm{p}$, we need to encode that the next time step $t+1$ corresponds to the unique $j$ such that $l_j=1$. That is, $\varphi^{k_\mathrm{p}} = \vee_{j=1}^{k_\mathrm{p}} (l_j \wedge \varphi_1^j)$, wherein $\varphi_1^j$ encodes the result of the query $(g_\mathrm{p}, v_i, j) \models \varphi_1$. It is straightforward to see that the latter constraint can also be encoded as mixed-integer constraints as $\vee$ can be transformed into $\wedge$ and $\lnot$, for which we already obtained mixed-integer constraints.

In the case $\varphi = \varphi_1 \mathcal{U} \varphi_2$, let $\varphi_1^j$ and $\varphi_2^j$ be the binaries or continuous variables encoding the result of $(g_\mathrm{p}, v_i, j) \models \varphi_1$ and $(g_\mathrm{p}, v_i, j) \models \varphi_2$, respectively, for all $j \in \{0,\hdots,k_\mathrm{p}\}$. If $t < k_\mathrm{p}$ and $\varphi^t$ is the result of $(g_\mathrm{p}, v_i, t) \models \varphi$, the definition of $\mathcal{U}$ enables to write that $\varphi^t = \varphi_2^t \vee (\varphi_1^t \wedge \varphi^{t+1})$. With a similar approach to~\cite{wolff2014optimization}, we resolve the circular reasoning appearing at $t = k_p$ by $\varphi^{k_\mathrm{p}}= \varphi^{k_\mathrm{p}}_2 \vee (\varphi^{k_\mathrm{p}}_1 \wedge (\vee_{j=1}^{k_\mathrm{p}} (l_j \wedge \varphi^{j}_\mathrm{c})))$,  where $\varphi^{j}_\mathrm{c}$ is recursively defined by $\varphi^{k_\mathrm{p}}_\mathrm{c} = \varphi^{k_\mathrm{p}}_2$,  and $\varphi^j_\mathrm{c} = \varphi^{j}_2 \vee (\varphi^{j}_1 \wedge \varphi^{j+1}_\mathrm{c})$ for all $0 \leq j < k_\mathrm{p}$. Thus, since equivalent mixed-integer constraints can be obtained for the $\vee$ and $\wedge$ operators, by induction we can also construct mixed-integer constraints for satisfiability of $\varphi$.

Finally, in the case $\varphi = \exists^{N}(\bigcirc \cdots \bigcirc)\varphi_1$, let denote by $\mathcal{S}_i$ the set of subset of $V$ such that 
    \begin{equation}
    \begin{aligned}
        \mathcal{S}_i = \{\{v_1,\hdots,v_k\}\: | &\: v_j \in \bigcirc\cdots \bigcirc(\{v_i\}), \: j \leq k, \\
                                        &  k \geq N, v_j \neq v_p, \text{ for } p \neq  j\}.
    \end{aligned}
    \label{neighbor_set_nodes}
    \end{equation}
    If $\varphi^t$ denotes the result of the query $(g_\mathrm{p}, v_i, t) \models \varphi$, then $\mathcal{S}_i$ is empty implies that $\varphi^t = 0$. Otherwise, if $\mathcal{S}_i$ is non empty, we have that $\varphi^t = \vee_{S \in \mathcal{S}_i}( \wedge_{v_k \in S}(\varphi^t_1 (v_k)))$, where $\varphi^t_1(v_k)$ encodes the result of the query $(g_\mathrm{p}, v_k, t) \models \varphi_1$. Thus, we can construct mixed-integer constraints for satisfaction of $\varphi$.

\begin{corollary}[\textsc{MILP for infinite-horizon GTL}]\label{eq:milp-spec}
    Given a GTL formula $\varphi$, a trajectory length $k_\mathrm{p}$, and a node $v_i \in V$, the existence of a periodic graph trajectoy $g_\mathrm{p}$ of length $k_\mathrm{p}$ such that $(g_\mathrm{p},v_i) \models \varphi$ can be equivalently formulated as the mixed-integer constraint $A^i [\boldsymbol{x}; w^{q_i}; l] \leq b^i$, where the variables are $\boldsymbol{x} = [x^1(0);\hdots; x^1(k_\mathrm{p}); \hdots; x^m(0);\hdots; x^m(k_\mathrm{p})]$, $w^{q_i} \in [0,1]^{q_i}$ has continuous and binary components, and $l \in \{0,1\}^{k_\mathrm{p}}$ is such that if there exists $j \in \{1,\hdots,k_\mathrm{p}\}$ with $l_j = 1$, then the resulting $g_\mathrm{p}$ is $(k_\mathrm{p}, j)$-periodic. Further, the parameters $b^i \in \mathbb{R}^{p_i}$, $A^i$ of appropriate dimensions, and $q_i ,p_i\in \mathbb{N}$ depend only on $\varphi$ and $v_i$.
\end{corollary}

\begin{remark}
    Note that due to page limitations, the MILP encoding of GTL formulas shortly described in this section might not be optimal in the obtained number of constraints and continuous components of $w^{q_i}$. However, our code implementation provides efficient encoding of $\vee$, $\mathcal{U}$, safety property $\square \varphi$, persistence $\Diamond \square \varphi$, and liveness $\square \Diamond \varphi$.
\end{remark}

\subsection{Synthesis of a Time-Varying Markov Matrix via MINLPs}

Corollary~\ref{eq:milp-spec} shows that the synthesis of a graph trajectory satisfying a GTL formula at a given node can be equivalently formulated as mixed-integer constraints. However, the resulting graph trajectory must also incorporate the dynamics of the sub-swarms given by~\eqref{eq:swarm_dyn} and their motion constraints.

\begin{lemma}[\textsc{General MINLP formulation}] \label{corr:minlp-formulation-constr}
    Let $G = (V,E)$ be the labeled graph induced by the topology of the bins as in Definition~\ref{def:gtl-swarm}, $\varphi$ be a GTL formula, $k_\mathrm{p} \in \mathbb{T}$ be the desired length of a periodic graph trajectory, $V'$ be a subset of $V$, $\mathcal{C}$ be the cost function to minimize, and $x^1(0),\hdots, x^m(0)$ be the initial density distributions for all the sub-swarms. Then, the following statements are equivalent:
    \begin{enumerate}
        \item There exists a periodic graph trajectory $g_\mathrm{p}$ of length $k_\mathrm{p}$ such that $(g_\mathrm{p} , v_i) \models \varphi$ for all $v_i \in V'$ while the motion constraints are satisfied and the cost $\mathcal{C}$ is minimized.
        \item There exists a solution to the MINLP~\eqref{eq:cost-minlp}--\eqref{pgtl_markov_state_evol}.
    \end{enumerate}
        \begin{align}
            & \underset{x^s,M^s,w^{q_i}, l}{\mathrm{minimize}} \quad \quad \quad \quad \quad  \sum_{t=0}^{k_\mathrm{p}} \mathcal{C}((x^s(t))^m_{s=1},(M^s(t))^m_{s=1}) \label{eq:cost-minlp}\\
            & \mathrm{subject \ to} \quad \quad \quad \quad \quad \quad w^{q_i} \in [0,1]^{q_i}, l \in \{0,1\}^{k_\mathrm{p}}\nonumber\\
            &\forall v_i \in V', \quad \quad \quad \quad  \quad \:  \quad A^i [\boldsymbol{x}; w^{q_i}; l] \leq b^i, \label{pgtl_induced_constr}\\
            &\forall t,s \in \mathbb{N}_{[1,k_\mathrm{p}] \times [1,m]}, \quad \quad  \boldsymbol{1}^{\textrm{T}} \: x^s(t) = 1, \label{pgtl_dist_constr} \\
            & \forall t,s \in \mathbb{N}_{[1,k_\mathrm{p}] \times [1,m]} , \quad \quad x^s(t) \geq \boldsymbol{0}, \label{pgtl_pos_x_constr} \\
            & \forall t,s \in \mathbb{N}_{[0,k_\mathrm{p}-1] \times [1,m]}, \quad \boldsymbol{1}^{\textrm{T}} \: M^s(t) = \boldsymbol{1}^{\textrm{T}}, \label{pgtl_stochas_constr}\\
            &\forall t,s \in \mathbb{N}_{[0,k_\mathrm{p}-1] \times [1,m]}, \quad M^s(t) \geq \boldsymbol{0}, \label{pgtl_pos_P_constr} \\
            &\forall t,s \in \mathbb{N}_{[0,k_\mathrm{p}-1] \times [1,m]}, (\boldsymbol{1}\boldsymbol{1}^{\textrm{T}} - (A^s_{\mathrm{adj}})^{\textrm{T}}) \odot M^s(t) = \boldsymbol{0}, \label{pgtl_adjacency_constr}\\
            & \forall t,s \in \mathbb{N}_{[0,k_\mathrm{p}-1] \times [1,m]}, \quad x^s(t+1) = M^s(t) \: x^s(t),\label{pgtl_markov_state_evol}
        \end{align}
        where the variables  are $x^s(t)$, $M^s(t)$, $l \in \{0,1\}^{k_\mathrm{p}}$, $w^{q_i} \in [0,1]^{q_i}$ for all $v_i \in V'$ with $i$ denoting a bin index, the parameters $A^i,b^i$ and $q_i$ depend only on $v_i$ and $\varphi$ for all $v_i \in V'$, and $\mathbb{N}_{[a,b] \times [c,d]} = \{a,\hdots,b\} \times \{c,\hdots,d\}$. Recall that $\boldsymbol{x}$ stacks $x^s(t)$ for all possible values of $s$ and $t$, and $w^{q_i}$ has components that can be either binary and continuous.
\end{lemma}

\begin{proof}
    This is a direct application of Corollary~\ref{eq:milp-spec}. The constraint~\eqref{pgtl_induced_constr} is obtained by the equivalence shown in Corollary~\ref{eq:milp-spec}. The bilinear constraint~\eqref{pgtl_markov_state_evol}, source of nonlinearity, is resulting from the dynamics~\eqref{eq:swarm_dyn}. The definition of the density distribution, the stochasticity, and the motion constraints are given by~\eqref{pgtl_dist_constr}--\eqref{pgtl_adjacency_constr}. 
\end{proof}

\section{Efficient Solutions} \label{sec:efficient-solution}
In this section, we provide an efficient algorithm to find locally-optimal solutions to the mixed-integer nonlinear program (MINLP)~\eqref{eq:cost-minlp}--\eqref{pgtl_markov_state_evol}. We first show that for a specific and widely-used subset of GTL specifications, it is only sufficient to solve a linear program (LP) or semi-definite program (SDP). Then, we use the specific structure of the problem to propose an efficient sequential mixed-integer linear programming (MILP) to address the problem.

\subsection{LP and SDP Formulations for Reach-Avoid Specifications}
We first specify explicitly the subset of GTL formulas corresponding to reach-avoid specifications.
\begin{definition}[\textsc{Reach-avoid specifications}]\label{def:reach-avoid-spec}
    Given the distributions $\nu^1, \hdots, \nu^m \in [0,1]^{n_\mathrm{r}}$, a GTL formula $\varphi$ encoding safety constraints (i.e., avoid specifications) in the form $\varphi = \wedge_{k=1}^{n_\mathrm{s}} (\square (\pi_k))$, where $\pi_k$ is an atomic node proposition, the reach-avoid specifications constrain the densities $x^1,\hdots, x^m$ to reach the steady-state distributions $\nu^1, \hdots, \nu^m$, respectively, while the resulting graph trajectory must satisfy $\varphi$.
\end{definition}
\begin{remark}
    The steady-state distribution constraints can be also encoded using GTL formulas with operators such as $\Diamond \square  \pi$ or $\Diamond \pi$, where $\pi$ is applied on an adequate node labelling.
\end{remark} 

Recall that in order for a graph trajectory to satisfy an atomic proposition $\pi$ at a node $v_i \in V$ and time $t \in \mathbb{T}$, we have that $g(v_i,t) \in \mathcal{O}(\pi)$. That is, there should exist $x^s(t)$ for all $s \in \{1,\hdots,m\}$ such that $A f_i(x^1(t),\hdots,x^m(t)) \leq b$, where $A$ and $b$ are defined by the polyhedra $\mathcal{O}(\pi)$. Since the function $f_i$ is affine in its arguments, we can write such constraint as the linear constraint $A^i [x^1(t);\hdots;x^m(t)] \leq b^i$, where $A^i$ of adequate dimension encodes both $A$ and the linear part of $f_i$ and $b^i \in \mathbb{R}^{p_i}$ incorporates both $b$ and the constant part of $f_i$.

As a consequence, the satisfiability of a safety specification, e.g. $\square (\pi)$ at $v_i$, can be equivalently formulated as the infinite-dimensional linear constraint 
\begin{align}\label{eq:safety-constraint}
    A^i [x^1(t);\hdots;x^m(t)] \leq b^i, \: \forall t \geq 0,
\end{align} 
where $x^s(t)$ for all $s \in \{1,\hdots,m\}$ are the variables.
\begin{lemma}[\textsc{Finite-dimensional linear encoding for safety constraints}]\label{lem:linear-constraint-safety}
    Assume $A^i [x^1(0);\hdots;x^m(0)] \leq b^i$ is satisfied. Then, the safety specification given by the infinite-dimensional constraint~\eqref{eq:safety-constraint} is satisfied if and only if there exists $Y \in \mathbb{R}^{ p_i \times p_i}$, $S \in \mathbb{R}^{p_i \times m}$  such that 
    \begin{align} 
	        Y b^i + S \boldsymbol{1}   &\geq -b^i, \label{eq:safety-constraint-finite-1}\\
	        Y A^i + S \mathbb{O} &\leq -A^i \mathcal{M}(t), \label{eq:safety-constraint-finite-2}\\
	        Y &\leq \boldsymbol{0}, \label{eq:safety-constraint-finite-3}
    \end{align}
    where $\mathcal{M}(t) = \mathrm{diag}(M^1(t),\hdots, M^m(t))$ is a block diagonal matrix of $M^s(t)$ for all $s \in \{1,\hdots,m\}$, the matrix $\mathbb{O} \in \mathbb{R}^{m \times n_\mathrm{r}m}$ satisfies $\mathbb{O}_{i,j} = 1$ for all $i \in \{1,\hdots,m\}$, $j \in \{n_\mathrm{r}(i-1), \hdots, n_\mathrm{r} i\}$, and $\mathbb{O}_{i,j} = 0$ otherwise. 
\end{lemma} 
\begin{proof}
    Let define the set $\mathcal{Y}$ of distributions characterizing the safety constraints by 
    $$\mathcal{Y} = \{y = [y^1;\hdots;y^m] \in \mathbb{R}^{n_\mathrm{r}m } | y \geq 0, \mathbb{O} y =  \boldsymbol{1}, A^i y \leq b^i \}.$$
    Let $x(t) = [x^1(t);\hdots;x^m(t)]$. Since the safety constraint~\eqref{eq:safety-constraint} is satisfied at $t=0$, one can observe that it remains satisfied if and only if $\forall x(t) \in \mathcal{Y}, A^i x(t+1) = A^i \mathcal{M}(t)x(t) \leq b^i$. The latter condition holds if and only for all $k \in \{1,\hdots,p_i\}$
    \begin{align}
        &\mathrm{maximize} \{ e_k^\mathrm{T} A^i \mathcal{M}(t)x(t) | x(t) \in \mathcal{Y}\} \leq b^i_k \label{eq:safety-ctr-max} \\
        \Longleftrightarrow &\mathrm{minimize} \{ - e_k^\mathrm{T} A^i \mathcal{M}(t)x(t) | x(t) \in \mathcal{Y}\} \geq -b^i_k. \nonumber
    \end{align}
    In the standard form, the minimization problem is given by
    \begin{equation}\label{proof:safety-min-lp-sf}
        \begin{aligned}
            & \underset{v = [x; s] \geq 0}{\mathrm{minimize}} & & [ - e_k^\mathrm{T} A^i \mathcal{M}(t) \: \: \boldsymbol{0}] v \\
            & \mathrm{subject \ to} & &  \begin{bmatrix} A^i & \mathbb{I} \\ \mathbb{O} & \boldsymbol{0} \end{bmatrix} v = \begin{bmatrix} b^i \\ \boldsymbol{1} \end{bmatrix},
        \end{aligned}
    \end{equation}
    where $\mathbb{I}$ is the identity matrix of appropriate dimension. Thus, the dual form of the above LP standard form is given by
    \begin{equation*}
	    \begin{aligned}
	        & \underset{y_k, s_k}{\mathrm{maximize}} & & [(b^i)^\mathrm{T} \: \: \boldsymbol{1}^\mathrm{T}] \begin{bmatrix} y_k \\ s_k \end{bmatrix} \\
	        & \mathrm{subject \ to} & &  \begin{bmatrix} (A^i)^\mathrm{T} & \mathbb{O}^\mathrm{T} \\ \mathbb{I} & \boldsymbol{0} \end{bmatrix} \begin{bmatrix} y_k \\ s_k \end{bmatrix} \leq \begin{bmatrix} -(A^i \mathcal{M}(t))^\mathrm{T} e_k \\ \boldsymbol{0} \end{bmatrix},
	    \end{aligned}
	\end{equation*}
	for all $k \in \{1,\hdots,p_i\}$.  For an LP, Strong duality holds when either the primal or dual problem is feasible~\cite{boyd2004convex}. Since the constraint~\eqref{eq:safety-constraint} is satisfied at $t=0$, $\mathcal{Y}$ is non-empty. Thus, the primal~\eqref{proof:safety-min-lp-sf} is feasible, hence strong duality holds. As a consequence, the constraint given by~\eqref{eq:safety-ctr-max} is equivalent to the existence of  $(y_k^*,s_k^*) \in \mathbb{R}^{p_i} \times \mathbb{R}^{m}$ such that
	\begin{equation*}
	    \begin{aligned}
	        (b^i)^\mathrm{T} y_k^* + \boldsymbol{1}^\mathrm{T} s_k^*   &\geq -b^i_k, \\
	        (A^i)^\mathrm{T} y_k^* + \mathbb{O}^\mathrm{T} s_k^* &\leq -(A^i \mathcal{M}(t))^\mathrm{T} e_k, \\
	        y_k^* &\leq \boldsymbol{0},
	    \end{aligned}
	\end{equation*}
	for all $k \in \{1,\hdots,p_i\}$. Hence, the announced result.
\end{proof}

\begin{remark}
    Lemma~\ref{lem:linear-constraint-safety} states that as long as there exists $Y$, $S$ and $\mathcal{M}(t)$ such that the linear constraints~\eqref{eq:safety-constraint-finite-1}--\eqref{eq:safety-constraint-finite-3} are satisfied, we only need to consider time-invariant markov matrices, i.e. $\mathcal{M} = \mathrm{diag}(M^1,\hdots, M^m)$, to satisfy the safety constraints~\eqref{eq:safety-constraint}. Hence, the linear constraints have finite dimension. In the remainder of this section, we focus on such time-invariant matrices and denote $\mathcal{M}$ instead of $\mathcal{M}(t)$.
\end{remark}

With the safety specifications given by~\eqref{eq:safety-constraint} expressed as finite-dimensional linear constraints, we now focus on the reach specifications. 
\begin{lemma}[\textsc{Ergodicity constraint,~\cite{seneta2006non}}]\label{lem:erg-constraint}
    Assume that each graph $G^s = (V, E^s)$ of each sub-swarm is strongly connected, i.e there exists a path between every pair of bins $\mathcal{R}_i$ and $\mathcal{R}_j$. Then, $\nu^1, \hdots, \nu^m$ are steady-state distributions of the sub-swarms if and only if the desired time-invariant markov matrices $M^s$ for all $s \in \{1,\hdots,m\}$ satisfy
    \begin{align}
        M^s \nu^s = \nu^s, \quad \forall s \in \{1,\hdots,m\}. \label{eq:erg-constraint}
    \end{align}
\end{lemma}

Although Lemma~\ref{lem:erg-constraint} enables to write the reach specifications as the linear constraint~\eqref{eq:erg-constraint}, we seek for Markov matrices that converge optimally to the steady-state distribution.

\begin{definition}[\textsc{Coefficient of ergocity},~\cite{seneta2006non}]\label{def:coeff-ergocity}
    For a stochastic matrix $M \in \mathbb{R}^{n \times n}$, its coefficient of ergocity $\tau_1 (M)$ is defined by $\tau_1(M) = 0.5\: \underset{i,j \in \{1,\hdots,n\}}{\mathrm{max}} \sum_{p=1}^n |M_{p,i} - M_{p,j}|.$
\end{definition} 

\begin{lemma}[\cite{seneta2006non}, Theorem~2.10] \label{lem:max_bound_rate} 
    Given a stochastic matrix $M$, suppose $\lambda$ is an eigenvalue of $M$ such that $\lambda \neq 1$. Then, $|\lambda| \leq \tau_1 (M)$. In particular, the rate of convergence to the steady-state distribution given by the second largest eigen value, $\lambda_2(M)$, is such that $|\lambda_2(M)| \leq \tau_1(M)$.
\end{lemma}

As a consequence of Lemma~\ref{lem:max_bound_rate}, by minimizing the linear function $\tau_1(M^s)$, one can obtain a tight upper bound on the rate of convergence of $x^s(t)$ to  $\nu^s$. 
\begin{definition}[\textsc{Scrambling pattern}]
    A graph $G^s = (V, E^s)$ has a scrambling pattern if for every pair of rows $i , j$, there exists a column $k$ such that $A^s_\mathrm{adj}[i,k] = A^s_\mathrm{adj}[j,k] = 1$.
\end{definition}    

We demonstrate in Lemma~\ref{lem:scrambling-pattern} that when the graph associated to each sub-swarm has a scrambling pattern, minimizing the function $\tau_1(M^s)$ guarantees an exponential rate of convergence to the steady-state distribution since $\tau_1(M^s)$ is an upper bound on the second largest eigen value $\lambda_2(M^s)$.

\begin{lemma}[\textsc{Exponential convergence rate for graphs  with scrambling pattern}]\label{lem:scrambling-pattern}
    The graph $G^s = (V, E^s)$ has a scrambling pattern for all $s \in \{1,\hdots,m\}$ if and only if  $\tau_1(M^s) < 1$. Hence $\lambda_2(M^s) < 1$ and we ensure an exponential convergence rate to the steady-state distribution.
\end{lemma}
\begin{proof}
    By the scrambling pattern, for all $i,j \in \{1 \hdots n_\mathrm{r}\}$ there exists $k \in \{1,\hdots,n_\mathrm{r}\}$ such that $A^s_\mathrm{adj}[i,k] = A^s_\mathrm{adj}[j,k] = 1$. Without loss of generality, consider that $M^s \geq \epsilon A^s_\mathrm{adj}$ for some small fixed $\epsilon > 0$. That is, $M^s$ preserves the connectivity of $G^s$. Therefore, we have that $M^s_{k,i} \geq \epsilon$ and $M_{k,j} \geq \epsilon$ and 
    \begin{align*}
        \tau_1(M^s) &\leq 0.5  \sum_{p=1}^{n_\mathrm{r}} |M_{p,i} - M_{p,j}| \\
                  &= 0.5 (|M^s_{k,i} - M^s_{k,j}| + \sum_{p \neq k} |M^s_{p,i} - M^s_{p,j}|)\\
                  &\leq 0.5 (|M^s_{k,i} - M^s_{k,j}| + \sum_{p \neq k} M^s_{p,i} + \sum_{p \neq k} M^s_{p,j}) \\
                  &= 0.5 (|M^s_{k,i} - M^s_{k,j}| + 2 - (M^s_{k,i} + M^s_{k,j})) \\
                  & = \begin{cases}
                            1 - M^s_{k,j} , & \mathrm{if }\: M^s_{k,i} - M^s_{k,j} \geq 0 \\
                            1 -  M^s_{k,i} , & \mathrm{otherwise}
                       \end{cases} \\
                  & \leq 1 - \epsilon < 1.
    \end{align*}
    Thus, since $\lambda_2(M^s) \leq \tau_1(M^s) < 1$, we have exponential convergence to the steady-state distribution. On the other hand, if $A^s_\textrm{adj}$ does not have the scrambling pattern, there exists $i_0,j_0 \in \{1 \hdots n_{\mathrm{r}}\}$ such that for all $k \in \{1 \hdots n_{\mathrm{r}}\}$ either $M^s_{k,i_0} = 0$ and $M^s_{k,j_0} > 0$ or $M^s_{k,i_0} > 0$ and $M^s_{k,j_0} = 0$. As a consequence, $\sum_{k=1}^{n_{\mathrm{r}}} |M^s_{k,i_0} - M^s_{k,j_0}| = 2$ which is the maximum possible value that can be attained by $\sum_{p=1}^{n_\mathrm{r}} |M_{p,i} - M_{p,j}|$ for all $i,j \in \{1 \hdots n_{\mathrm{r}}\}$. Hence $\tau_1(M^s) = 1$, and the equivalence is therefore obtained.
\end{proof} 

To summarize, consider reach-avoid specifications encoded with $A^i$ and $b^i$ as detailed in~\eqref{eq:safety-constraint}, nonnegative weights $c_1, \hdots, c_m$ specifying the relative importance of the rate of convergence of each sub-swarm, and a cost function $\mathcal{C} : (\mathbb{R}^{n_\mathrm{r} \times n_\mathrm{r}} )^m \mapsto \mathbb{R}$ to be optimized. Then, the time-invariant Markov matrices $M^1, \hdots, M^s$ solution of the LP~\eqref{eq:cost-lp}--\eqref{eq:pos-M-lp} induce, through the evolution of $x^1(t),\hdots,x^m(t)$, a graph trajectory that satisfies the GTL formula.
\begin{align}
    & \underset{M^s, S, Y}{\mathrm{minimize}} & & \mathcal{C}((M^s)^m_{s=1}) + \sum_{s=1}^m c_s \tau_1(M^s) \label{eq:cost-lp}\\
    & \mathrm{subject \ to} & & Y \leq \boldsymbol{0}, \label{eq:lp-sto}\\
    & & & Y b^i + S \boldsymbol{1}   \geq -b^i, \\
    & & & Y A^i + S \mathbb{O} \leq -A^i \mathrm{diag}(M^1,\hdots,M^s),\\
    & \forall s \in \{1,\hdots,m\}, & & (\boldsymbol{1}\boldsymbol{1}^{\mathrm{T}} - (A^s_{\mathrm{adj}})^{\mathrm{T}}) \odot M^s = \boldsymbol{0},  \\
    & \forall s \in \{1,\hdots,m\}, & & \boldsymbol{1}^\mathrm{T} M^s = \boldsymbol{1}^\mathrm{T}, \\
    & \forall s \in \{1,\hdots,m\}, & &  M^s \nu^s = \nu^s, \\
    & \forall s \in \{1,\hdots,m\}, & & M^s \geq \boldsymbol{0}. \label{eq:pos-M-lp}
\end{align}

\begin{remark}
    Note that without the scrambling assumption on $G^s = (V, E^s)$, we have that $\tau_1(M^s) = 1$. Hence, optimizing the ergocity coefficient does not guarantee any convergence to the desired distribution $\nu^s$. In this scenario, we propose, as follows, an SDP formulation to control the rate of convergence.
\end{remark} 

\begin{lemma}[\textsc{Convergence rate for graphs  without scrambling pattern}]\label{lem:noscrambling-pattern}
    Assume that the graph $G^s = (V, E^s)$ has no scrambling pattern for all $s \in \{1,\hdots,m\}$. Let $\mathcal{M}(M^s) = M^s diag(\nu^s) (M^s)^\mathrm{T} diag(\nu^s)^{-1}$ be the \textit{multiplicative reversiblization}~\cite{fill1991eigenvalue} of $M^s$. Then, we have that:
    \begin{enumerate}
        \item The rate of convergence of $M^s$ to $\nu^s$ is given by $\lambda_2(\mathcal{M}(M^s))$ and exponential when $\lambda_2(\mathcal{M}(M^s)) < 1$.
        \item $\lambda_2(\mathcal{M}(M^s)) = ||(Q^s)^{-1} M^s Q^s - r^s (r^s)^\mathrm{T}||^2_2$, where $r^s=\sqrt{\nu^s}$, and $Q^s = diag(r^s)$.
    \end{enumerate}
\end{lemma}
\begin{proof}
    First, from Theorem $2.7$ in~\cite{fill1991eigenvalue}, we have that \begin{equation*}
        4||x^s(n) -v^s||^2 \leq \lambda_2(\mathcal{M}(M^s))^n (\chi^s_0)^2,
    \end{equation*}
    where $x^s(n)$ is the Markov chain state distribution at time index $n$, and $\chi^s_0 = \sum_k \frac{(x^s_k(0)- v^s_k)^2}{v^s_k}$. As an immediate consequence, the rate of convergence of $M^s$ to $\nu^s$ is given by $\lambda_2(\mathcal{M}(M^s))$ and exponential when $\lambda_2(\mathcal{M}(M^s)) < 1$.
    
    Second, we characterize the second eigen value of $\mathcal{M}(M^s)$.
    \begin{align*}
        \mathcal{M}(M^s) diag(\nu^s)   &= M^s diag(\nu^s) (M^s)^\mathrm{T} diag(\nu^s)^{-1} diag(\nu^s) \\
                 &= M^s diag(\nu^s) (M^s)^\mathrm{T}\\
                 &= diag(\nu^s) diag(\nu^s)^{-1} M^s diag(\nu^s) (M^s)^\mathrm{T} \\
                 &= diag(\nu^s) \mathcal{M}(M^s)^\mathrm{T}.
    \end{align*}
    
    By left and right multiplication of the equation above by $(Q^s)^{-1}$, we obtain
    \begin{equation*}
        (Q^s)^{-1} \mathcal{M}(M^s) Q^s = Q^s \mathcal{M}(M^s)^\mathrm{T} (Q^s)^{-1}.
    \end{equation*}
    We deduce that $(Q^s)^{-1} \mathcal{M}(M^s) Q^s$ is symmetric and has same eigenvalues as $\mathcal{M}(M^s)$. Moreover, the followings hold.
    \begin{align}
	    &\lambda_{\mathrm{max}}( (Q^s)^{-1} \mathcal{M}(M^s) Q^s) = \lambda_{\mathrm{max}}(\mathcal{M}(M^s)) = 1, \label{eq:lam-reversibility}\\
	    &(Q^s)^{-1} \mathcal{M}(M^s) Q^s r^s =  (Q^s)^{-1} \mathcal{M}(M^s) \nu^s = (Q^s)^{-1} \nu^s = r^s, \label{eq:lam-eigenright}\\
	   &||r^s||_2 = \sqrt{\sum_{i=1}^{n_\mathrm{r}}( r^s_i)^2} = \sqrt{\sum_{i=1}^{n_\mathrm{r}} \nu^s_i} = 1, \label{eq:lam-eigenunit}
	\end{align}
	where $\lambda_{\mathrm{max}}(\cdot)$ denotes the maximum eigen value. The last equality of equation~\eqref{eq:lam-reversibility} is due to $\mathcal{M}(M^s)$ straightforwardly being a stochastic matrix. The equation~\eqref{eq:lam-eigenright} comes from straightforward algebra and combined with~\eqref{eq:lam-eigenunit}, we have that $r^s$ is an unit eigenvector of $(Q^s)^{-1} \mathcal{M}(M^s) Q^s$ associated with the maximum eigenvalue $1$. Thus, a classic result in algebra linking the second eigen value and the maximum eigen value provides that
	\begin{equation} \label{proof:lam2-M}
	    \lambda_2(\mathcal{M}(M^s)) = \lambda_{\mathrm{max}}((Q^s)^{-1} \mathcal{M}(M^s) Q^s - r^s (r^s)^\mathrm{T}).
	\end{equation}
    Observe that 
    \begin{equation} \label{proof:sim-lam2-form}
        \begin{aligned}
            &((Q^s)^{-1} M^s Q^s - r^s (r^s)^\mathrm{T})((Q^s)^{-1} M^s Q^s - r^s (r^s)^\mathrm{T})^T  \\
            =& (Q^s)^{-1} M^s Q^s Q^s (M^s)^\mathrm{T} (Q^s)^{-1} - (Q^s)^{-1} M^s \nu^s (r^s)^\mathrm{T} \\
            & \quad - r^s (\nu^s)^\mathrm{T} (M^s)^\mathrm{T} (Q^s)^{-1} + r^s (r^s)^\mathrm{T} r^s (r^s)^\mathrm{T} \\
            =& (Q^s)^{-1} \mathcal{M}(M^s) Q^s - r^s (r^s)^\mathrm{T} - r^s (r^s)^\mathrm{T} + r^s (r^s)^\mathrm{T}  \\
            =& (Q^s)^{-1} \mathcal{M}(M^s) Q^s - r^s (r^s)^\mathrm{T}. 
        \end{aligned}
    \end{equation}
    By combining~\eqref{proof:lam2-M} and~\eqref{proof:sim-lam2-form}, we finally have that
    \begin{align*}
        \lambda_2(\mathcal{M}(M^s)) &= \lambda_{max}((Q^s)^{-1} \mathcal{M}(M^s) Q^s - r^s (r^s)^\mathrm{T}) \\
                        &= ||(Q^s)^{-1} M^s Q^s - r^s (r^s)^\mathrm{T}||_2^2.
    \end{align*}
\end{proof}
\begin{remark}
    Lemma~\ref{lem:noscrambling-pattern} provides a way to control the rate of convergence to the stationary distribution via the \emph{convex} function $||(Q^s)^{-1} M^s Q^s - r^s (r^s)^\mathrm{T}||_2^2$ of $M^s$. Note that if $M^s$ is a reversible Markov matrix~\cite{boyd2004fastest}, we have that $||(Q^s)^{-1} M^s Q^s - r^s (r^s)^\mathrm{T}||_2^2 = \lambda_\mathrm{max}((Q^s)^{-1} M^s Q^s - r^s (r^s)^\mathrm{T})$, which is widely studied in the literature of the fastest mixing rate for Markov chains~\cite{boyd2004fastest}. However, in this paper, we characterize the rate of convergence via the convex $||(Q^s)^{-1} M^s Q^s - r^s (r^s)^\mathrm{T}||_2^2$ without assuming reversibility of $M^s$.
\end{remark}   

As a consequence of Lemma~\ref{lem:noscrambling-pattern},  the time-invariant Markov matrices $M^1, \hdots, M^s$ solutions of the SDP~\eqref{eq:cost-sdp} induce, through the evolution of $x^1(t),\hdots,x^m(t)$, a graph trajectory that satisfies the specifications.
\begin{equation}
\begin{aligned}
    & \underset{M^s, S, Y}{\mathrm{minimize}} & & \mathcal{C}((M^s)^m_{s=1}) + \sum_{s=1}^m c_s ||(Q^s)^{-1} M^s Q^s - r^s (r^s)^\mathrm{T}||_2^2 \label{eq:cost-sdp}\\
    & \mathrm{subject \ to} & & \eqref{eq:lp-sto}-\eqref{eq:pos-M-lp}.
\end{aligned}
\end{equation}

\noindent{\textbf{Complexity and correctness analysis}.} The worst-case time complexity to solve the LP~\eqref{eq:cost-lp}-\eqref{eq:pos-M-lp} and SDP~\eqref{eq:cost-sdp} is polynomial in its number of constraints and variables. Specifically, we have $m n_\mathrm{r}^2 + p_i^2 + p_i m$ number of variables, where we recall that $p_i$ is defined as the number of constraints enforced by the specifications. Note that $p_i$ is therefore proportional to the size of the specifications. Similarly, the number of constraints can be straightforwardly upper-bounded by $p_i^2 + p_i + n_\mathrm{r} m p_i + 2 m n_\mathrm{r}^2 + 2m n_\mathrm{r}$. Besides, by Lemma~\ref{lem:linear-constraint-safety}, a solution of the LP or SDP ensures the satisfaction of the constraints. Therefore, the algorithm for reach-avoid specifications is correct.

\subsection{Special Case: MILP Formulation for Complete Graphs} \label{subsec:spec-case}

In the scenario where the GTL formula $\varphi$ does not express reach-avoid specifications, if each $G^s = (V,E^s)$ is a complete graph, we can reduce the MINLP feasibility problem given by constraints~\eqref{pgtl_induced_constr}--\eqref{pgtl_markov_state_evol} to a MILP feasibility problem.

\begin{corollary}[\textsc{MILP for complete graphs}] \label{corr:complete_graph_case_milp}
   With the notation of Lemma~\ref{corr:minlp-formulation-constr}, assume that each graph $G^s$ is complete. Then, the latter statements are equivalent:
    \begin{enumerate}
        \item There exists a periodic graph trajectory $g_\mathrm{p}$ of length $k_\mathrm{p}$ such that $(g_\mathrm{p} , v_i) \models \varphi$ for all $v_i \in V'$ while the motion constraints are satisfied.
        \item There exists an optimal solution to the MILP problem given by~\eqref{eq:cost-minlp}--\eqref{pgtl_pos_x_constr}.
    \end{enumerate}
    Furthermore, if there exists $\hat{x}^s(t)$ for all $(s,t) \in \mathbb{N}_{[1,m] \times [1,k_\mathrm{p}]}$  satisfying constraints~\eqref{pgtl_induced_constr}--\eqref{pgtl_pos_x_constr}, then $\hat{M}^s(t)$ given by
    \begin{equation} \label{p_sol_comple_graph}
        \hat{M}^s_{i,j}(t) = \hat{x}^s_i(t+1), \; \forall i,j \in \{1,\hdots,n_\mathrm{r}\},
    \end{equation}
    for $t \in \{0,\hdots,k_\mathrm{p}-1\}$, satisfies constraints~\eqref{pgtl_stochas_constr}--\eqref{pgtl_markov_state_evol}. 
\end{corollary}
\begin{proof}
    For a complete graph $G^s$, we have that $A^s_{\mathrm{adj}} = \boldsymbol{1}\boldsymbol{1}^T$. Thus, the constraint~\eqref{pgtl_adjacency_constr} is automatically satisfied. Further, $1)$ implies $2)$ is trivial by the equivalence of Corollary \ref{corr:minlp-formulation-constr}. 
    
    Suppose $2)$ is valid, i.e. the constraints~\eqref{pgtl_induced_constr}--\eqref{pgtl_pos_x_constr} are satisfied by $\hat{x}^s(t)$ for all $s$ and $t$. We want to show the the constraints~\eqref{pgtl_stochas_constr}--\eqref{pgtl_markov_state_evol} are automatically satisfied. With $\hat{M}^s(t)$ given by~\eqref{p_sol_comple_graph} and $j \in \{1,\hdots,n_\mathrm{r}\}$,
    \begin{equation*}
        \textstyle \sum_{i=1}^{n_\mathrm{r}} \hat{M}^s_{i,j}(t) = \sum_{i=1}^{n_\mathrm{r}} \hat{x}^s_i(t+1) = \boldsymbol{1}^{\textrm{T}} \hat{x}^s(t+1) = 1.
    \end{equation*}
    This yields the satisfiability of the constraint~\eqref{pgtl_stochas_constr} by $\hat{M}^s(t)$. The constraint~\eqref{pgtl_pos_P_constr} is satisfied by $\hat{M}^s(t)$ due to the constraint~\eqref{pgtl_pos_x_constr}. Finally, for $i \in \{1,\hdots,n_\mathrm{r}\}$, we have
    \begin{equation*}
        \textstyle \sum_{j=1}^{n_\mathrm{r}} M^s_{i,j}(t) \hat{x}^s_j(t) = \hat{x}^s_i(t+1) \sum_{j=1}^{n_\mathrm{r}} \hat{x}^s_j(t) = \hat{x}^s_i(t+1).
    \end{equation*}
    Thus, we have the satisfiability of the constraint~\eqref{pgtl_markov_state_evol} by $\hat{M}^s(t)$. Hence, $2)$ implies $1)$ as the constraints~\eqref{pgtl_induced_constr}--\eqref{pgtl_markov_state_evol} are satisfied by $\hat{x}^s(t)$ and $\hat{M}^s(t)$ for all $(s,t) \in \mathbb{N}_{[1,m] \times [1,k_\mathrm{p}]}$.
\end{proof}

In the case of complete graphs, Corollary~\ref{corr:complete_graph_case_milp} provides that the feasibility of constraints~\eqref{pgtl_induced_constr}--\eqref{pgtl_markov_state_evol} is equivalent to the feasibility of~\eqref{pgtl_induced_constr}--\eqref{pgtl_pos_x_constr}. Thus, when the cost function $\mathcal{C}$ is a function of only the densities $x^1(t),\hdots,x^m(t)$, we can rewrite the MINLP optimization problem~\eqref{eq:cost-minlp}--\eqref{pgtl_markov_state_evol} as the MILP problem~\eqref{eq:cost-minlp}--\eqref{pgtl_pos_x_constr}. The resulting densities are then used in~\eqref{p_sol_comple_graph} to find the Markov matrices. When the cost function $\mathcal{C}$ is dependent of $M^s(t)$, one can obtain suboptimal solutions by replacing $M^s(t)$ with the corresponding $x^s(t)$ as in~\eqref{p_sol_comple_graph}.\vspace{0.10cm}

\noindent{\textbf{Complexity and correctness analysis.}} With the notation of Corollary~\ref{corr:minlp-formulation-constr}, the number $N_\textrm{c}$ of non-binary variables and an upper bound $N_\textrm{b}$ on the number of binary variables of the equivalent MILP in Corollary~\ref{corr:complete_graph_case_milp} are given by $N_\textrm{c}=n_{\textrm{r}} k_\mathrm{p} m$ and $N_\textrm{b} = k_\mathrm{p} + \sum_{v_i \in V'} q_i$. The number of constraints $C$ of the MILP is $C=N_\textrm{c}+\sum_{v_i \in V'} p_i$, where $q_i$ is the dimension of $b_i$. Since a linear program (LP) can be solved in polynomial time in the number of variables and constraints via interior-point methods~\cite{nesterov1994interior}, the worst-case time complexity to solve the MILP is $O(2^{N_\textrm{b}}R(N_\textrm{c},C))$, $R$ is a polynomial. By Corollary~\ref{corr:complete_graph_case_milp}, a solution to the MILP ensures the satisfaction of the constraints. Therefore, the algorithm for the special case is correct by construction.

\subsection{General Case: Trust-Region-Based Sequential Mixed-Integer Programming}
In this section, we make no assumptions on the structure of the GTL specifications and the graph of each sub-swarm. Then, we develop an efficient sequential mixed-integer linear programming scheme to solve the MINLP~\eqref{eq:cost-minlp}--\eqref{pgtl_pos_x_constr}.\vspace{0.10cm}

\noindent{\textbf{Linearizing the nonconvex constraints.}} The idea of the efficient solving scheme is to reduce the problem to an adequate set of MILPs that can be solved efficiently and optimally by off-the-shell solvers. Specifically, we solve the nonconvex problem by sequentially linearizing the constraint~\eqref{pgtl_markov_state_evol} around the solution of the $k^{\mathrm{th}}$ iteration. This linearization results into a MILP. The obtained solutions are then used for the $(k+1)^{\mathrm{th}}$ iteration. We begin by denoting the solutions of the $k^{\mathrm{th}}$ iteration by $x^{s,k}(t)$ and $M^{s,k}(t-1)$ for all $t\in \{1,\hdots,k_p\}$ and $s\in \{1,\hdots,m\}$. Thus, at the $(k+1)^\mathrm{th}$ iteration, the first-order approximation of $x^s(t+1) = M^s(t) \: x^s(t)$ around the previous solutions $x^{s,k}(t)$ and $M^{s,k}(t)$ is given by
\begin{align}
    x^{s}(t+1) = &x^{s,k}(t)M^{s,k}(t)  + M^{s,k}(t) \big( x^{s}(t)-x^{s,k}(t) \big) \nonumber\\
    &+ \big( M^s(t) - M^{s,k}(t) \big) x^{s,k}(t),\label{eq:first-order-approx}
\end{align}
where $x^{s,k}(0) = x^s(0)$ for all iteration $k$. First, note that the linearization~\eqref{eq:first-order-approx} may create an infeasible problem. To mitigate the effects of this infeasibility, we augment the linearized dynamics with the unconstrained slack variable $z^{s}(t) \in \mathbb{R}^{n_\mathrm{r}}$. Thus, the resulting constraint is always feasible and can be written as follows:
\begin{align}
    x^{s}(t+1) = x^{s,k}(t+1) &+ M^{s,k}(t) \big( x^{s}(t)-x^{s,k}(t) \big) \label{linearize-nonconvex}\\
    &+ \big( M^s(t) - M^{s,k}(t) \big) x^{s,k}(t) + z^s(t). \nonumber
\end{align}
Further, to ensure that the variable $z^s(t)$ is used only when necessary, we augment the cost function with a sufficiently large penalization weight $\lambda > 0$. Thus, the solution for the $(k+1)^\mathrm{th}$ iteration optimizes the linearized cost given by
 \begin{align}
     L^k(\boldsymbol{x},\boldsymbol{M}) =  \mathrm{cost}(\boldsymbol{x}, \boldsymbol{M})+ \lambda \sum_{t=0}^{k_\mathrm{p}-1} \sum_{s=1}^m \|z^s(t)\|, 
 \end{align}
 where $\mathrm{cost}(\boldsymbol{x},\boldsymbol{M}) = \sum_{t=0}^{k_\mathrm{p}}  \mathcal{C}((x^s(t))^m_{s=1},(M^s(t))^m_{s=1})$ and $\|\cdot\|$ can be either the infinity norm or $1$-norm. Recall that $\boldsymbol{x} \in \mathbb{R}^{n_\mathrm{r} \times m \times k_\mathrm{p}}$ contains the densities of all sub-swarms at all time and similarly we define $\boldsymbol{M}$ to contain the Markov matrices at all time and for all sub-swarms.\vspace{0.10cm}
 
 \noindent{\textbf{Trust Region Constraints and linearized problem.}} We ensure that the resulting density distribution $x^s(t)$ at the $(k+1)^\mathrm{th}$ iteration does not deviate significantly from the density obtained at the $k^\mathrm{th}$ iteration by imposing, for all $s \in \{1,\hdots,m\}$ and $t \in \{1,\hdots,k_\mathrm{p}\}$, the following trust region constraint 
\begin{align}
    \|x^s(t) - x^{s,k}(t)\| \leq r^k, \label{eq:trust-region}
\end{align}
where $r^k$ is a trust region that will be updated at each iteration so that the solution $x^s(t)$ remains close to the density obtained in the previous iteration, $x^{s,k}(t)$. This update rule enables to keep the solutions within a region where the linearization is accurate. As a consequence, at the $(k+1)^\mathrm{th}$ iteration, the convex subproblem is given by
\begin{align}
        & \underset{x^s,M^s,w^{q_i}, l, z^s}{\mathrm{minimize}} \quad \quad \quad \quad \quad  L^k(\boldsymbol{x}, \boldsymbol{M}) \label{eq:cost-minlp_c}\\
        & \mathrm{\ \ subject \ to} \quad \quad \quad \quad \quad \: \: \eqref{pgtl_induced_constr}-\eqref{pgtl_adjacency_constr},\eqref{linearize-nonconvex},~\eqref{eq:trust-region} \nonumber
\end{align}

\noindent{\textbf{Starting point via McCormick relaxations.}} The choice of the starting points $x^{s,0}(t)$ and $M^{s,0}(t)$ are crucial to accelerate and provide feasible solutions to the MINLP problem. We seek to get as close as possible feasible and optimal solutions. To this end, we write the constraint $x^s(t+1) = M^s(t) x^s(t)$ component-wise as $x^s_i(t+1) = \sum_{j=1}^{n_\mathrm{r}} V^s_{i,j}(t)$, where $V^s_{i,j}(t) = M^s_{i,j}(t) x^s_j(t)$ is a new variable. Then, for all $s \in \{1,\hdots, m\}$, $t \in \{0,\hdots,k_\mathrm{p}-1\}$, and $i,j \in \{1,\hdots,n_\mathrm{r}\}$, we have the following McCormick relaxation of the bilinear constraint $V^s_{i,j}(t) = M^s_{i,j}(t) x^s_j(t)$:
\begin{align}
    V^s_{i,j}(t) \geq 0,& \quad V^s_{i,j}(t) \geq M^s_{i,j}(t) + x^s_j(t) - 1,\label{eq:mccormick-1}\\
     V^s_{i,j}(t) \leq x^s_j(t),&  \quad V^s_{i,j}(t) \leq M^s_{i,j}(t).\label{eq:mccormick-2}
\end{align}
As a consequence, the starting point $x^{s,0}(t)$ of Algorithm~\ref{alg:scp-minlp} is an optimal solution of the relaxed MILP problem
\begin{align}
        & \underset{x^s,M^s,w^{q_i}, l, z^s, V^s}{\mathrm{minimize}} \quad \quad \quad  \: \: \: \mathrm{cost}(\boldsymbol{x}, \boldsymbol{M}) \label{eq:init_problem}\\
        & \: \mathrm{\ \ subject \ to} \quad \quad \quad \quad \quad \eqref{pgtl_induced_constr}-\eqref{pgtl_adjacency_constr},\eqref{eq:mccormick-1}-\eqref{eq:mccormick-2} \nonumber\\
        & \forall t,s \in \mathbb{N}_{[0,k_\mathrm{p}-1] \times [1,m]}, \quad x^s(t+1) = V^s(t) \boldsymbol{1}.\label{eq:mccormick-changes}
\end{align}
Further, from the obtained starting point $x^{s,0}(t)$ for all $s$ and $t$, we seek for $M^{s,0}(t)$ that minimizes the error of not satisfying the bilinear constraints~\eqref{pgtl_markov_state_evol}. To this end, $M^{s,0}(t)$ is an optimal solution of the following LP problem
\begin{align}
    & \underset{M^s,z^s}{\mathrm{minimize}} \quad \quad \quad  \quad \quad \: \: \: \: \sum_{s=1}^m \sum_{t=0}^{k_\mathrm{p}-1}  ||z^s(t)||_1 \label{eq:accurate-M}\\
    &\: \mathrm{\ \ subject \ to} \quad \quad \quad \quad \quad \eqref{pgtl_stochas_constr}-\eqref{pgtl_adjacency_constr} \nonumber\\
    & \forall t,s \in \mathbb{N}_{[0,k_\mathrm{p}-1] \times [1,m]}, \quad x^{s,0}(t+1) = M^{s}(t) x^{s,0}(t) + z^s(t), \nonumber
\end{align}
where we relax the bilinear constraints and add a slack variable such that we penalize its use in the cost function.\vspace{0.10cm}

\noindent{\textbf{Sequential mixed-integer programming algorithm.}} Algorithm~\ref{alg:scp-minlp} summarizes the trust-region-based sequential convex optimization scheme to compute approximate (possibly local) solutions of~\eqref{eq:cost-minlp}--\eqref{pgtl_markov_state_evol}. Specifically, the quality of the solution is established using three metrics: The change $\Delta L^{k+1}$ in the optimal cost, the accuracy $\mathrm{f}^{k+1}$ of the bilinear constraint attained by the new solution, and the ratio $\rho^{k+1}$ of the resulting accuracy and past accuracy. These metrics are given by
\begin{align}
    \Delta L^{k+1} &= |\hat{L}^{k+1} -  \hat{L}^k|,\label{eq:resol-cost-form}\\
    \mathrm{f}^{k+1} &= \sum_{s=1}^m \sum_{t=0}^{k_\mathrm{p}-1}  ||x^{s,k+1}(t) - M^{s,k+1}(t) x^{s,k+1}(t)||_1,\label{eq:acc-form}\\
    \rho^{k+1} &= \mathrm{f}^{k+1} / \mathrm{f}^{k},\label{eq:rho-form}
\end{align} 
where $\hat{L}^k$ is the optimal cost of the linearized problem at iteration $k$. The ratio $\rho^{k+1}$ compares the accuracy of the new solution and the solution obtained at the past iteration. When $\rho^k > 1$, the new solution is considered inaccurate. Then, we contract the trust region $r^k$ and restart the iteration. If not, the solutions $M^{s,k+1}(t)$ and $x^{s,k+1}(t)$ are considered acceptable. Then, we move to the next iteration and expand the trust region depending on the value of $\rho^k$. Algorithm~\ref{alg:scp-minlp} stops when the minimum trust region value is reached or the cost cannot be improved while the bilinear constraint is satisfied with $\epsilon_\mathrm{acc}$.\vspace{0.10cm}

 \algdef{SE}[DOWHILE]{Do}{doWhile}{\algorithmicdo}[1]{\algorithmicwhile\ #1}%
\begin{algorithm}[!t]
    \caption{Sequential convex programming with trust region to efficiently solve the MINLP Problem~\eqref{eq:cost-minlp}--\eqref{pgtl_markov_state_evol}.}\label{alg:scp-minlp}
    \begin{algorithmic}[1]    
    \Require{Swarm distribution $x^s(0)$ for all $s\in\{1,\hdots,m\}$, penalty weight $\lambda > 0$, parameters $r_\mathrm{min}<1$, $r_\mathrm{exp}, r_\mathrm{con} > 1$, cost tolerance $\epsilon_{\mathrm{tol}} > 0$, and accuracy tolerance $\epsilon_{\mathrm{acc}} > 0$. }
    \Ensure{ $M^s(t)$ locally optimal solution of ~\eqref{eq:cost-minlp}--\eqref{pgtl_markov_state_evol} }
        \State Initialize $x^{s,0}(t)$ by solving MILP~\eqref{eq:init_problem}--\eqref{eq:mccormick-changes}
        \State Initialize $M^{s,0}(t)$ by solving LP~\eqref{eq:accurate-M}
        \State $k \gets 0$ and $r^k \gets 2$ \Comment{Initial trust region $r^0$}
        \Do
            \State Find $x^{s,k+1}(t)$ by solving~\eqref{eq:cost-minlp_c} at $x^{s,k}(t), M^{s,k}(t), r^k$
            \State Find $M^{s,k+1}(t)$ solution of ~\eqref{eq:accurate-M} \Comment{$x^{s,0} \gets x^{s,k+1}$ }
            \State Compute $\Delta L^{k+1}$, $\mathrm{f}^{k+1}$, $\rho^{k+1}$ from~\eqref{eq:resol-cost-form},~\eqref{eq:acc-form},~\eqref{eq:rho-form}\label{eq:lin-acc-quant}
            \If{$\Delta L^k \leq \epsilon_{\mathrm{tol}}$ and ($\mathrm{f}^k \leq \epsilon_\mathrm{acc}$ or $\mathrm{f}^{k+1} \leq \epsilon_\mathrm{acc}$)} 
                \State \Return $M^{s,k}(t)$ \Comment{Found a solution}
            \EndIf
            \If{$\rho^k > 1$} \Comment{New solution reduces accuracy}
                \State $r^k \gets r^k / \min \{r_\mathrm{con}, \rho^k\} $\Comment{Contract trust region}
            \Else \Comment{Accept new solution}
                \State $k \gets k + 1$ \Comment{Update estimate}
                \State $r^k \gets r^{k-1} \min \{1/\rho_k, r_\textrm{exp}\}$ \Comment{Expand trust region}
            \EndIf
        \doWhile{$r^k > r_\mathrm{min}$} \Comment{Minimum trust region reached}
    \State \Return $M^{s,k}(t)$
  \end{algorithmic}
\end{algorithm}

\textbf{Complexity and correctness analysis}: We consider in this analysis the notation of Corollary~\ref{corr:minlp-formulation-constr}. Let $N_\mathrm{iter}$ be the number of iterations required by Algorithm~\ref{alg:scp-minlp} to terminate. By arguments similar to the complete graph case and using the notation in the discussion of its complexity analysis, the worst-case time complexity of Algorithm~\ref{alg:scp-minlp} is $O(2^{N_\textrm{b}}R(N_\textrm{c},C)N_{\mathrm{iter}})$. According to~\cite{mao2019successive}, such a sequential convex optimization can achieve a linear rate of convergence. Besides,  a solution returned by Algorithm~\ref{alg:scp-minlp} ensures the satisfaction of the specification.

\subsection{The Complete Algorithm: \controlalgo{}}

We develop \controlalgo{} to compute the desired Markov matrices $M^s(t)$ for all $s\in \{1,\hdots,m\}$ and $t \in \{0,\hdots,k_\mathrm{p}-1\}$. \controlalgo{} chooses the most efficient and scalable formulation for the problem depending on whether the specifications are reach-avoid specifications, the graph has a scrambling pattern, the graph is complete, or none of these special cases holds. Algorithm~\ref{alg:gtlproco} provides the description of \controlalgo{}. The user should provide as an input to \controlalgo{} the trajectory length $k_\mathrm{p}$. However, one might derive a sequential algorithm with increasing length $k_\mathrm{p}$ as long as the problem is infeasible until a feasible solution can be found. 

The Markov matrices computed by \controlalgo{} are distributed to each agent in order for them to choose their bin-to-bin transitions. Algorithm~\ref{alg:psg-algo} is a decentralized algorithm that specifies how each agent probabilistically computes its target bin at each time index in order for the high-level task specifications to be satisfied.

\begin{remark}
    Note that using a finite number of agent $N^s_\mathrm{a}$ for the sub-swarm $s$, achieving exactly a desired density $x^s(t)$ might not be possible due to the quantization error $\frac{1}{N^s_\mathrm{a}}$. For example, if $x^s(t) = [\frac{1}{3},\frac{2}{3}]$ and $N^s_\mathrm{a} = 10$, the realized density by $N^s_\mathrm{a}$ is $[0.3, 0.7]$ due to the finite value of $N^s_\mathrm{a}$.
\end{remark}

 \algdef{SE}[DOWHILE]{Do}{doWhile}{\algorithmicdo}[1]{\algorithmicwhile\ #1}%
\begin{algorithm}[!t]
    \caption{\controlalgo{}: Find the Markov matrices $M^s(t)$ for all $s\in \{1,\hdots,m\}$ and $t \in \{0,\hdots,k_\mathrm{p}-1\}$ solutions of Problem~\ref{prob:mc-approach-gtl}.}\label{alg:gtlproco}
    \begin{algorithmic}[1]    
    \Require{Graph $G^s=(V,E^s)$, swarm distribution $x^s(0)$, GTL formula $\varphi$, a set $V'$ of nodes, a length $k_{\mathrm{p}}$, and the cost $\mathcal{C}$. }
    \Ensure{ $M^s(t)$ solution of Problem~\ref{prob:mc-approach-gtl} }
        \If{$\varphi$ expresses reach-avoid specifications}
            \If{all graphs $G^s$ have scrambling pattern}
                \State Compute $M^s(t)$ by solving LP~\eqref{eq:cost-lp}--\eqref{eq:pos-M-lp}
            \Else
                \State Compute $M^s(t)$ by solving SDP~\eqref{eq:cost-sdp}
            \EndIf
        \Else 
            \State Find MILP encoding from $\varphi$ and $V'$ via Corollary~\ref{eq:milp-spec}
            \If{all graphs $G^s$ are complete}
                \State Find $M^s(t)$ via MILP~\eqref{eq:cost-minlp}--\eqref{pgtl_pos_x_constr} and Corollary~\ref{corr:complete_graph_case_milp}
            \Else
                \State Find $M^s(t)$ via Algorithm~\ref{alg:scp-minlp}
            \EndIf
        \EndIf
    \State \Return $M^{s}(t)$
  \end{algorithmic}
\end{algorithm}

\begin{algorithm}[!t]
    \caption{Probabilistic Swarm Guidance for each agent in the sub-swarm $s \in \{1,\hdots,m\}$.}  
    \label{alg:psg-algo}
    \begin{algorithmic}[1]
        \State Identify the current bin $\mathcal{R}_i$
        \State Query $M^s_{ji}(t)$, for all $j \in \{1,\hdots,n_{\textrm{r}}\}$
        \State Generate $z$ from the uniform distribution on $[0,1]$
        \State Select bin $\mathcal{R}_j$ such that
        $\sum_{l=1}^{j-1} M^s_{li}(t) \leq z \leq \sum_{l=1}^{j} M^s_{li}(t)$
        \State Transit to bin $\mathcal{R}_j$ while achieving collision avoidance. \label{alg-local-interaction}
    \end{algorithmic}
\end{algorithm} 

\section{Numerical experiments} \label{sec:numercal-example}
\begin{figure*}
    \centering
    \input{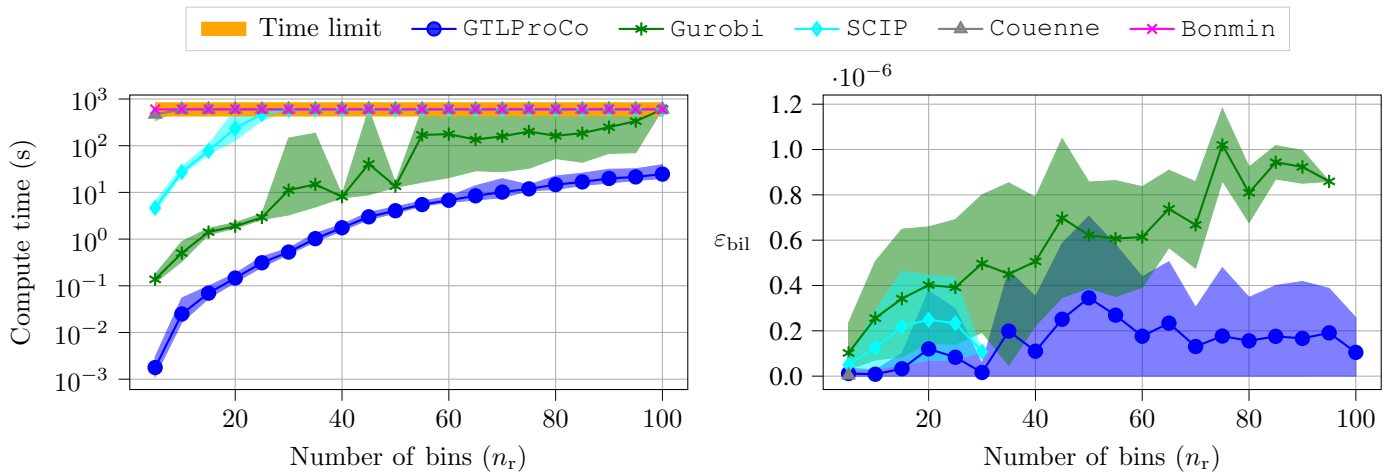}
    \vspace{-0.80cm}
    \caption{From left to right, we plot the computational time on a logarithm scale of each MINLP solver and the accuracy of the bilinear constraint of the solutions returned by each solver. We demonstrate that \controlalgo{} is significantly faster and more accurate than \texttt{Gurobi}, \texttt{SCIP}, \texttt{Couenne}, and \texttt{Bonmin}.}\vspace{-0.40cm}
    \label{fig:expMINLP}
\end{figure*}
In this section, we evaluate the \controlalgo{} on several swarm control tasks expressed using GTL. 
Specifically, we first empirically demonstrate that \controlalgo{} significantly improves scalability over off-the-shelf MINLP solvers applied on the original MINLP~\eqref{eq:cost-minlp}--\eqref{pgtl_markov_state_evol}. 
Besides, we show that \controlalgo{} can compute Markov matrices with a higher accuracy for the bilinear constraints than off-the-shelf MINLP solvers.
Second, in several gridworld examples, we show that our control approach is fast, sound, and can be applied in scenarios involving a large number of agents. 

All the experiments of this paper are performed on a computer with an Intel Core $i9$-$9900$ CPU $3.1$GHz $\times 16$ processors and $31.2$ Gb of RAM. 
All the implementations are written and tested in Python $3.8$. We use Gurobi $9$~\cite{gurobi} to solve all the linear and mixed-integer linear programs in this paper. 
We use Mosek~\cite{mosek} to solve the semi-definite programs presented in this paper. We provide in \texttt{https://github.com/wuwushrek/GTLProCo} all the codes for reproducibility and the videos of the experiments.\vspace*{-0.25cm}

\subsection{Comparisons with Off-The-Shelf MINLP Solvers}

In this section, we  compare \controlalgo{} with open-source and efficient MINLP solvers such as \texttt{Gurobi}~\cite{gurobi}, \texttt{Couenne}~\cite{belotti2009couenne}, \texttt{Bonmin}~\cite{bonami2007bonmin}, and \texttt{SCIP}~\cite{GamrathEtal2016ZR}. 
To this end, we randomly generate both problem instances of different sizes and GTL specifications as follows:
\begin{itemize}
    \item We generate $2000$ random problem instances.
    \item We generate each problem instance such that the number of bins  $n_\mathrm{r} \in \{5, 10, 15, 20,25,\hdots, 90, 95, 100\}$, the trajectory length $k_\mathrm{p} \in \{5,6,6,7,7,\hdots,14,14,15\}$, and the number of sub-swarms $m = 1$.
    \item We generate the underlying graph $G$ for each problem instance such that each node in $G$ has a random number of edges between $2$ and $5$.
    \item We randomly generate the GTL formula of each problem such that the atomic propositions are random. 
    We use a set of operators in the list $\wedge$, $\vee$, $\square$, and $\Diamond$.
    We make sure that the MILP encoding of the GTL formula is feasible.  
    \item In each problem, since the $l_j$ terms for the loop constraints are variables, we impose the additional cost function $\mathcal{C}^\mathrm{loop} = \sum_{j=1}^{k_\mathrm{p}} (j+1) l_j$. 
    Basically, by minimizing the cost, we desire the time loop to start as early as possible.
\end{itemize}
We use the default parameters of each MINLP solver except for the time limit that we constrain to be \emph{ten minutes}.
For \controlalgo{}, we choose the trust region contraction and expansion parameters as $r_\mathrm{con} = 1.5$ and $r_\mathrm{exp} = 1.5$.
We also choose the minimum trust region value to be $r_\mathrm{min} = 1e^{-4}$, the linearization penalty to be $\lambda = 10$, the cost tolerance to be $\epsilon_{\mathrm{tol}} = 1e^{-6}$, and the accuracy tolerance to be $\epsilon_{\mathrm{acc}} = 1e^{-6}$.

We compare the computation time and the error of the bilinear constraint of \controlalgo{} with the off-the-shelf MINLP solvers over the aforementioned randomly-generated problem instances. 
Given a solution $x^s(t)$ and $M^s(t)$ of the MINLP problem, we define the error of the bilinear constraint as
\begin{align*}
    \varepsilon_{\mathrm{bil}} = \max_{s \in \{1,\hdots,m\}} \max_{t\in \{0,\hdots, k_\mathrm{p}-1\}} \|x^s(t+1) - M^s(t)x^s(t)\|_\infty.
\end{align*}%
Figure~\ref{fig:expMINLP} empirically demonstrates the superior performance of our control algorithm \controlalgo{}, both in terms of computation time and error of the bilinear constraint,  compared to off-the-shelf MINLP solvers. 
Specifically, it shows that \texttt{Bonmin} is unable to solve any problems in the given time limit while \texttt{Couenne} is only able to solve problems corresponding to $n_\mathrm{r} = 5$ with a computation time of $468s$.
\texttt{SCIP}  times out for problem instances with $n_\mathrm{r} \geq 25$ while \texttt{Gurobi} times out with $n_\mathrm{r} \geq 95$. 
Therefore, \texttt{Gurobi} is the only algorithm that achieves comparable performance with \controlalgo{}. 
The standard deviation in compute time demonstrates that the compute time of \controlalgo{} is more consistent in most examples compared to \texttt{Gurobi}. 
Finally, Figure~\ref{fig:expMINLP} additionally demonstrates that, in almost all cases, \controlalgo{} finds a solution with better accuracy for the bilinear constraint than the MINLP solvers.

\subsection{Homogeneous Swarm Subject to GTL Specifications}
\begin{figure*}
    \centering
    \includegraphics[width=1.75in,height=1.75in]{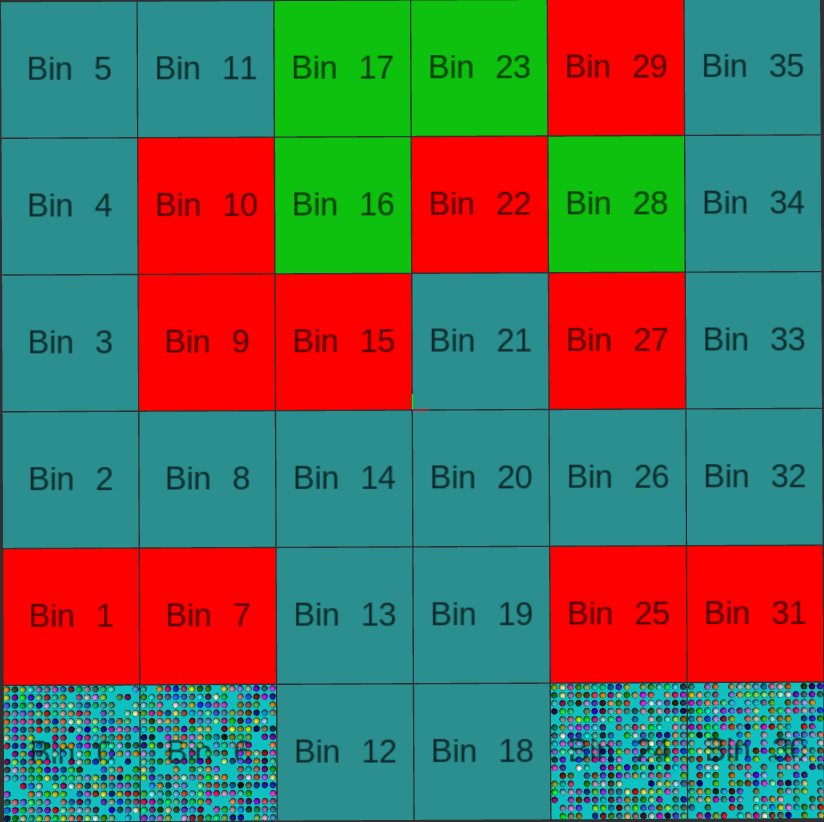}
    \includegraphics[width=1.75in,height=1.75in]{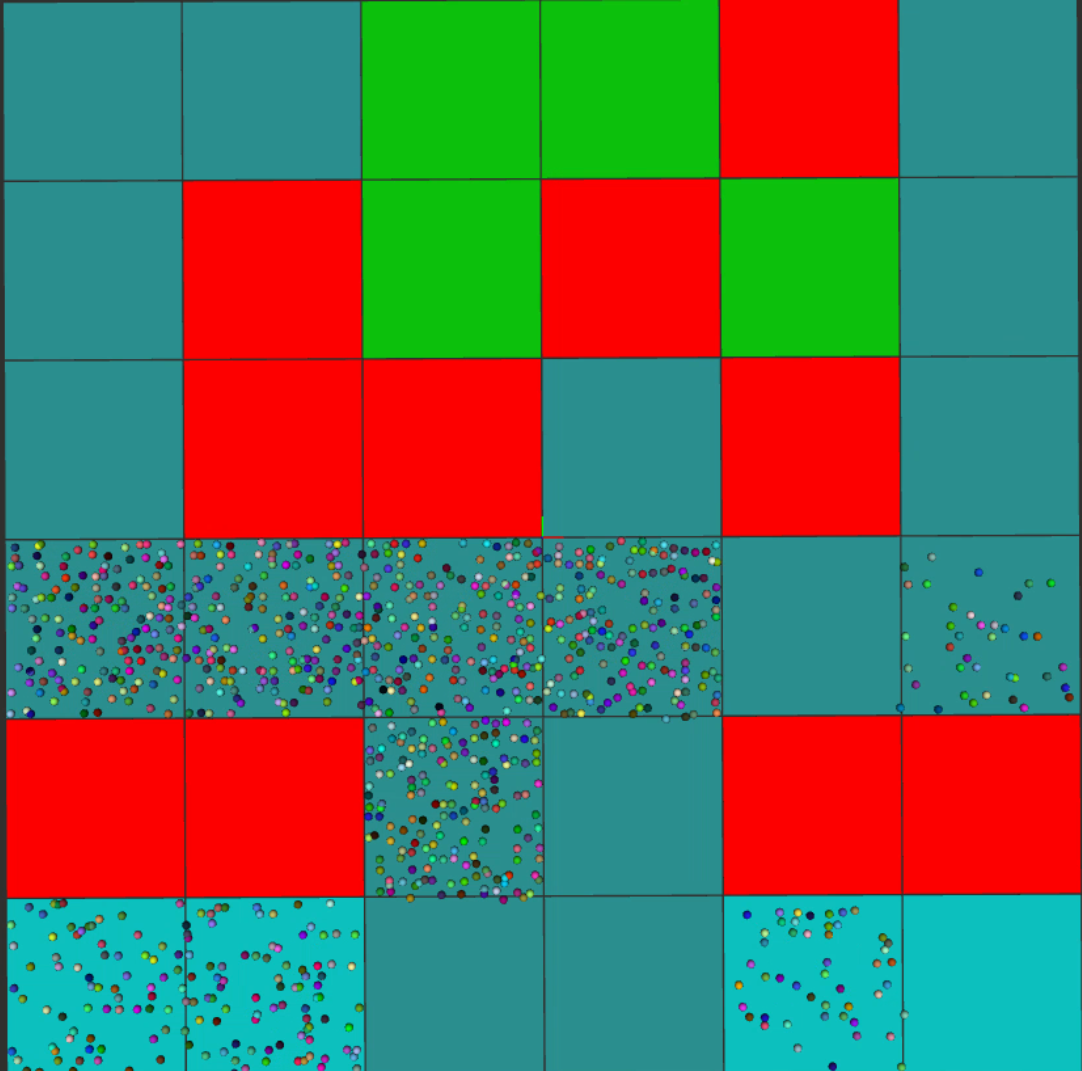}
    \includegraphics[width=1.75in,height=1.75in]{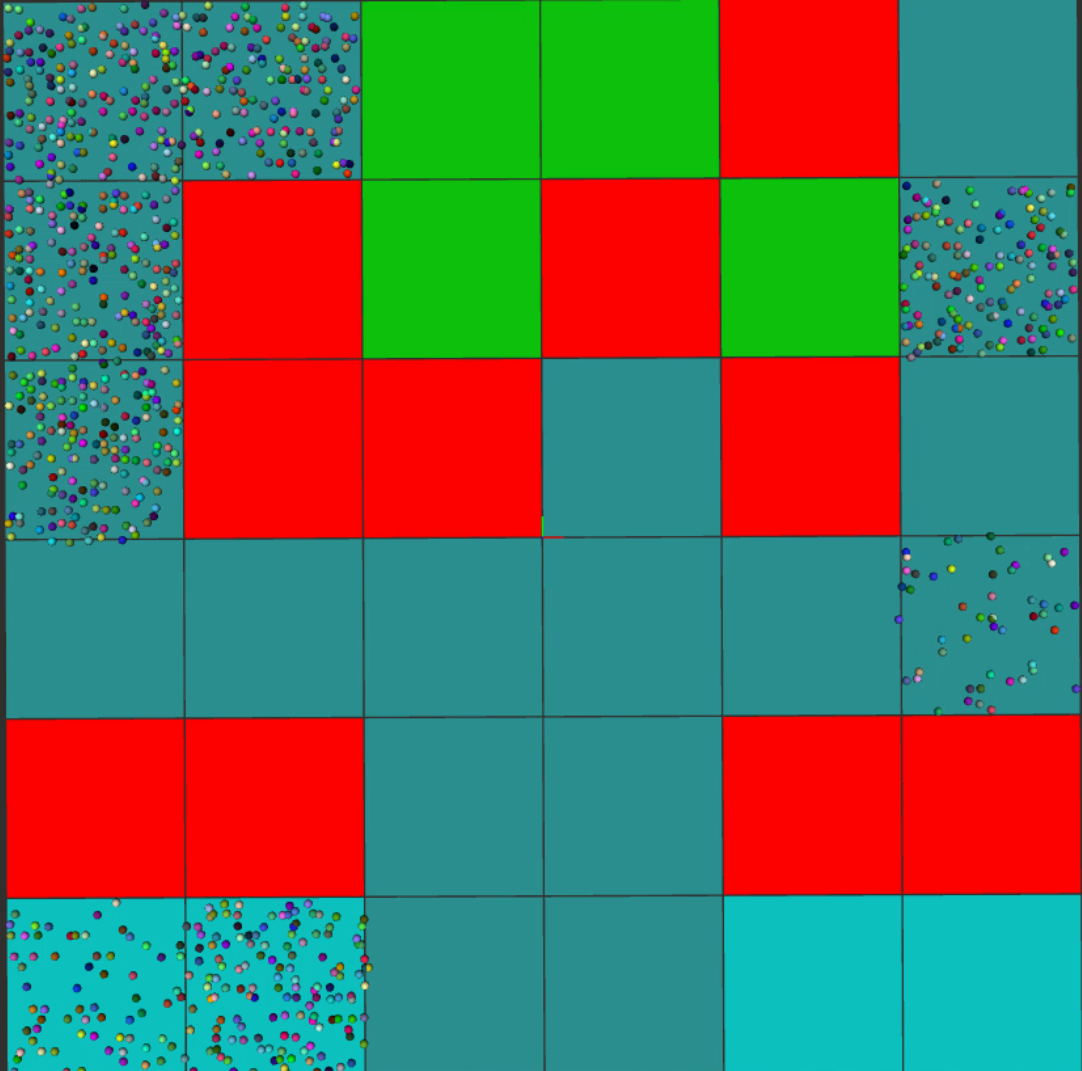}
    \includegraphics[width=1.75in,height=1.75in]{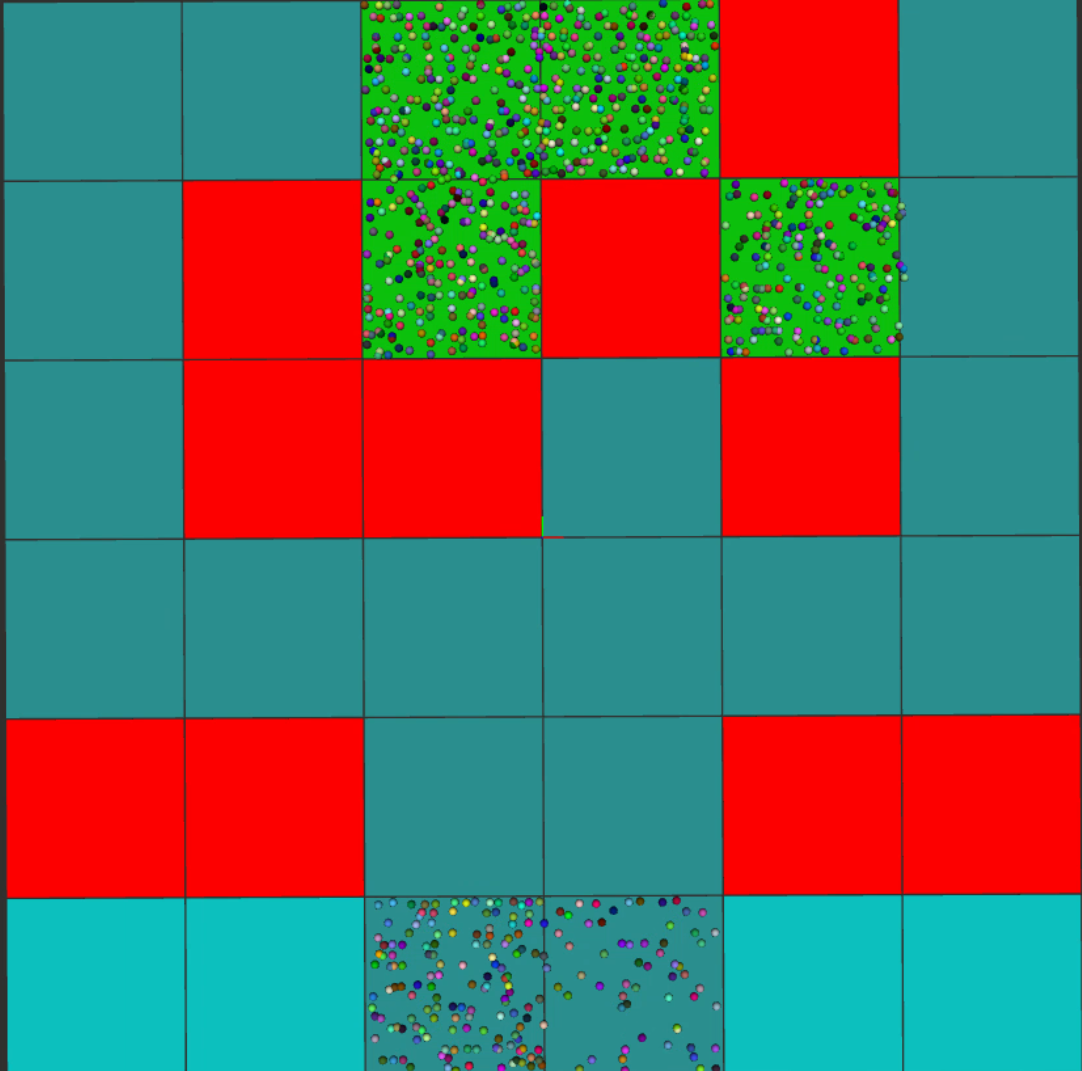}
    \vspace{-0.50cm}
    \caption{Evolution of the \emph{thousand} agents comprising the swarm at different time indexes. From the left to the right figure, we show the distribution of the swarm at time indexes $0$, $5$, $9$, and $15$. In the figure, the obstacle bins are in red, the target bins are in green, and the starting bins are in cyan.}\vspace{-0.40cm}
    \label{fig:evolution-swarm-homogeneous}
\end{figure*}
In this section, we consider a simulation example with a swarm of homogeneous agents navigating in a gridworld environment as shown in Figure~\ref{fig:evolution-swarm-homogeneous}. 
The desired behavior of the swarm is, from a given initial distribution, to reach a target distribution defined by some set of linear constraints while avoiding obstacles. 
Moreover, the swarm must satisfy some capacity constraints in the bins, i.e., each bin can contain only a fixed maximum number of agents at any time index. 
We obtain the graph representing the gridworld and label it with $f_{i}(x) = x$ for each node $i$. 
Then, we use GTL to express the task specifications of the homogeneous swarm ($m=1$) as following:\looseness=-1
\begin{itemize}
    \item Initially we have $x_i^1(0) = 0.25$ for all $i\in\{0,6,24,30\}$. That is, the agents are distributed in bins $0,6,24$, and $30$.
    \item We consider the GTL formulas $\Diamond \square (f_i \geq 0.2)$ for all $i \in \{16, 17, 23, 28\}$ and $\square (f_i = 0)$ for each obstacle bin $i$. Thus, we expect to swarm to reach final bins $16, 17, 23, 28$ while satisfying the specified constraints.
    \item We enforce the capacity constraints via the safe properties $\square (f_i \leq 0.25)$ for $i \in \{0,6,16,17,23,24,28,30\}$ and $\square (f_i \leq 0.15)$ for the remaining bins. 
    Thus, we relaxed the capacity constraints for the starting and final bins.
\end{itemize}

More specifically, we consider a scenario with a swarm comprised of $1000$ agents. 
We first apply \controlalgo{} to find a time-varying Markov matrix $M^1(t)$ such that the specifications above are satisfied. 
Then, at each time index, each agent independently and probabilistically chooses their target bin based on Algorithm~\ref{alg:psg-algo} with the computed $M^1(t)$.
\begin{figure}[!hbt]
    \centering
\begin{tikzpicture}

\definecolor{color0}{rgb}{0,1,1}
\definecolor{color1}{rgb}{1,0,1}

\begin{axis}[
width=8cm,
height=4cm,
legend cell align={left},
legend columns=5,
legend style={
  fill opacity=0.8,
  draw opacity=1,
  text opacity=1,
  at={(-0.15,1.25)},
  anchor=north west,
  draw=white!80!black
},
tick align=outside,
tick pos=left,
x grid style={white!69.0196078431373!black},
xlabel={Time steps},
xmajorgrids,
xmin=-1.45, xmax=30.45,
xtick style={color=black},
y grid style={white!69.0196078431373!black},
ylabel={Swarm density},
ymajorgrids,
ymin=-0.0114525, ymax=0.2405025,
ytick style={color=black},
ytick={-0.05,0,0.05,0.1,0.15,0.2,0.25},
yticklabels={−0.05,0.00,0.05,0.10,0.15,0.20,0.25}
]
\addplot [thick, red, mark=*, mark size=2, mark options={solid}]
table {%
0 0
1 0
2 0
3 0
4 0
5 0
6 0
7 0
8 0
9 0
10 0
11 0
12 0
13 0
14 0
15 0
16 0
17 0
18 0
19 0
20 0
21 0
22 0
23 0
24 0
25 0
26 0
27 0
28 0
29 0
};
\addlegendentry{\small Obstacles}
\addplot [thick, red, mark=*, mark size=2, mark options={solid}, forget plot]
table {%
0 0
1 0
2 0
3 0
4 0
5 0
6 0
7 0
8 0
9 0
10 0
11 0
12 0
13 0
14 0
15 0
16 0
17 0
18 0
19 0
20 0
21 0
22 0
23 0
24 0
25 0
26 0
27 0
28 0
29 0
};
\addplot [thick, red, mark=*, mark size=2, mark options={solid}, forget plot]
table {%
0 0
1 0
2 0
3 0
4 0
5 0
6 0
7 0
8 0
9 0
10 0
11 0
12 0
13 0
14 0
15 0
16 0
17 0
18 0
19 0
20 0
21 0
22 0
23 0
24 0
25 0
26 0
27 0
28 0
29 0
};
\addplot [thick, red, mark=*, mark size=2, mark options={solid}, forget plot]
table {%
0 0
1 0
2 0
3 0
4 0
5 0
6 0
7 0
8 0
9 0
10 0
11 0
12 0
13 0
14 0
15 0
16 0
17 0
18 0
19 0
20 0
21 0
22 0
23 0
24 0
25 0
26 0
27 0
28 0
29 0
};
\addplot [thick, red, mark=*, mark size=2, mark options={solid}, forget plot]
table {%
0 0
1 0
2 0
3 0
4 0
5 0
6 0
7 0
8 0
9 0
10 0
11 0
12 0
13 0
14 0
15 0
16 0
17 0
18 0
19 0
20 0
21 0
22 0
23 0
24 0
25 0
26 0
27 0
28 0
29 0
};
\addplot [thick, red, mark=*, mark size=2, mark options={solid}, forget plot]
table {%
0 0
1 0
2 0
3 0
4 0
5 0
6 0
7 0
8 0
9 0
10 0
11 0
12 0
13 0
14 0
15 0
16 0
17 0
18 0
19 0
20 0
21 0
22 0
23 0
24 0
25 0
26 0
27 0
28 0
29 0
};
\addplot [thick, red, mark=*, mark size=2, mark options={solid}, forget plot]
table {%
0 0
1 0
2 0
3 0
4 0
5 0
6 0
7 0
8 0
9 0
10 0
11 0
12 0
13 0
14 0
15 0
16 0
17 0
18 0
19 0
20 0
21 0
22 0
23 0
24 0
25 0
26 0
27 0
28 0
29 0
};
\addplot [thick, red, mark=*, mark size=2, mark options={solid}, forget plot]
table {%
0 0
1 0
2 0
3 0
4 0
5 0
6 0
7 0
8 0
9 0
10 0
11 0
12 0
13 0
14 0
15 0
16 0
17 0
18 0
19 0
20 0
21 0
22 0
23 0
24 0
25 0
26 0
27 0
28 0
29 0
};
\addplot [thick, red, mark=*, mark size=2, mark options={solid}, forget plot]
table {%
0 0
1 0
2 0
3 0
4 0
5 0
6 0
7 0
8 0
9 0
10 0
11 0
12 0
13 0
14 0
15 0
16 0
17 0
18 0
19 0
20 0
21 0
22 0
23 0
24 0
25 0
26 0
27 0
28 0
29 0
};
\addplot [thick, red, mark=*, mark size=2, mark options={solid}, forget plot]
table {%
0 0
1 0
2 0
3 0
4 0
5 0
6 0
7 0
8 0
9 0
10 0
11 0
12 0
13 0
14 0
15 0
16 0
17 0
18 0
19 0
20 0
21 0
22 0
23 0
24 0
25 0
26 0
27 0
28 0
29 0
};
\addplot [thick, blue, mark=asterisk, mark size=2, mark options={solid}]
table {%
0 0
1 0
2 0
3 0
4 0
5 0
6 0
7 0
8 0
9 0
10 0
11 0.114449977874756
12 0.0401500463485718
13 0.199949979782104
14 0.199949979782104
15 0.199949979782104
16 0.199949979782104
17 0.199949979782104
18 0.199949979782104
19 0.199949979782104
20 0.199949979782104
21 0.199949979782104
22 0.199949979782104
23 0.199949979782104
24 0.199949979782104
25 0.199949979782104
26 0.199949979782104
27 0.199949979782104
28 0.199949979782104
29 0.199949979782104
};
\addlegendentry{\small $\mathcal{R}_{16}$}
\addplot [thick, green!50.1960784313725!black, mark=diamond*, mark size=2, mark options={solid}]
table {%
0 0
1 0
2 0
3 0
4 0
5 0
6 0
7 0
8 0
9 0
10 0.151600003242493
11 0.148750066757202
12 0.229050040245056
13 0.201349973678589
14 0.197149991989136
15 0.201349973678589
16 0.197149991989136
17 0.201349973678589
18 0.197149991989136
19 0.201349973678589
20 0.197149991989136
21 0.201349973678589
22 0.197149991989136
23 0.201349973678589
24 0.197149991989136
25 0.201349973678589
26 0.197149991989136
27 0.201349973678589
28 0.197149991989136
29 0.201349973678589
};
\addlegendentry{\small $\mathcal{R}_{17}$}
\addplot [thick, color0, mark=triangle*, mark size=2, mark options={solid}]
table {%
0 0
1 0
2 0
3 0
4 0
5 0
6 0
7 0
8 0
9 0
10 0
11 0.0354499816894531
12 0.178550004959106
13 0.197149991989136
14 0.201349973678589
15 0.197149991989136
16 0.201349973678589
17 0.197149991989136
18 0.201349973678589
19 0.197149991989136
20 0.201349973678589
21 0.197149991989136
22 0.201349973678589
23 0.197149991989136
24 0.201349973678589
25 0.197149991989136
26 0.201349973678589
27 0.197149991989136
28 0.201349973678589
29 0.197149991989136
};
\addlegendentry{\small $\mathcal{R}_{23}$}
\addplot [thick, color1, mark=x, mark size=2, mark options={solid}]
table {%
0 0
1 0
2 0
3 0
4 0
5 0
6 0
7 0
8 0.0497499704360962
9 0
10 0
11 0.0580500364303589
12 0.0593999624252319
13 0.201550006866455
14 0.201550006866455
15 0.201550006866455
16 0.201550006866455
17 0.201550006866455
18 0.201550006866455
19 0.201550006866455
20 0.201550006866455
21 0.201550006866455
22 0.201550006866455
23 0.201550006866455
24 0.201550006866455
25 0.201550006866455
26 0.201550006866455
27 0.201550006866455
28 0.201550006866455
29 0.201550006866455
};
\addlegendentry{ $\mathcal{R}_{28}$}
\end{axis}

\end{tikzpicture}
\begin{tikzpicture}

\definecolor{color0}{rgb}{0,1,1}
\definecolor{color1}{rgb}{1,0,1}

\begin{axis}[
width=8cm,
height=4cm,
legend columns=4,
legend cell align={left},
legend style={fill opacity=0.8, draw opacity=1, text opacity=1, at={(0,1.25)},
  anchor=north west, draw=white!80!black},
tick align=outside,
tick pos=left,
x grid style={white!69.0196078431373!black},
xlabel={Time steps},
xmajorgrids,
xmin=-1.45, xmax=30.45,
xtick style={color=black},
y grid style={white!69.0196078431373!black},
ylabel={Swarm density},
ymajorgrids,
ymin=-0.012695, ymax=0.266595,
ytick style={color=black},
ytick={-0.05,0,0.05,0.1,0.15,0.2,0.25,0.3},
yticklabels={−0.05,0.00,0.05,0.10,0.15,0.20,0.25,0.30}
]
\addplot [thick, blue, mark=asterisk, mark size=2, mark options={solid}]
table {%
0 0.249300003051758
1 0.101600050926208
2 0.101600050926208
3 0.147850036621094
4 0.147850036621094
5 0.059499979019165
6 0
7 0.059499979019165
8 0.059499979019165
9 0.059499979019165
10 0.059499979019165
11 0
12 0
13 0
14 0
15 0
16 0
17 0
18 0
19 0
20 0
21 0
22 0
23 0
24 0
25 0
26 0
27 0
28 0
29 0
};
\addlegendentry{\small $\mathcal{R}_0$}
\addplot [thick, green!50.1960784313725!black, mark=diamond*, mark size=2, mark options={solid}]
table {%
0 0.250699996948242
1 0.246799945831299
2 0.196699976921082
3 0
4 0
5 0.0883500576019287
6 0.059499979019165
7 0
8 0.140499949455261
9 0.140499949455261
10 0
11 0.200000047683716
12 0.0577000379562378
13 0
14 0
15 0
16 0
17 0
18 0
19 0
20 0
21 0
22 0
23 0
24 0
25 0
26 0
27 0
28 0
29 0
};
\addlegendentry{\small $\mathcal{R}_6$}
\addplot [thick, color0, mark=triangle*, mark size=2, mark options={solid}]
table {%
0 0.246099948883057
1 0.238600015640259
2 0.206400036811829
3 0.0521500110626221
4 0
5 0.0521500110626221
6 0
7 0
8 0
9 0
10 0
11 0
12 0
13 0
14 0
15 0
16 0
17 0
18 0
19 0
20 0
21 0
22 0
23 0
24 0
25 0
26 0
27 0
28 0
29 0
};
\addlegendentry{\small $\mathcal{R}_{24}$}
\addplot [thick, color1, mark=x, mark size=2, mark options={solid}]
table {%
0 0.253900051116943
1 0.114699959754944
2 0
3 0
4 0.0521500110626221
5 0
6 0
7 0
8 0
9 0
10 0
11 0
12 0
13 0
14 0
15 0
16 0
17 0
18 0
19 0
20 0
21 0
22 0
23 0
24 0
25 0
26 0
27 0
28 0
29 0
};
\addlegendentry{\small $\mathcal{R}_{30}$}
\end{axis}

\end{tikzpicture}
\begin{tikzpicture}

\definecolor{color0}{rgb}{0,1,1}
\definecolor{color1}{rgb}{1,0,1}

\begin{axis}[
width=8cm,
height=4cm,
legend columns=4,
legend cell align={left},
legend style={
  fill opacity=0.8,
  draw opacity=1,
  text opacity=1,
  at={(0,1.25)},
  anchor=north west,
  draw=white!80!black
},
tick align=outside,
tick pos=left,
x grid style={white!69.0196078431373!black},
xlabel={Time steps},
xmajorgrids,
xmin=-1.45, xmax=30.45,
xtick style={color=black},
y grid style={white!69.0196078431373!black},
ylabel={Swarm density},
ymajorgrids,
ymin=-0.0077125, ymax=0.1619625,
ytick style={color=black},
ytick={-0.02,0,0.02,0.04,0.06,0.08,0.1,0.12,0.14,0.16,0.18},
yticklabels={−0.02,0.00,0.02,0.04,0.06,0.08,0.10,0.12,0.14,0.16,0.18}
]
\addplot [thick, blue, mark=asterisk, mark size=2, mark options={solid}]
table {%
0 0
1 0
2 0
3 0
4 0.151600003242493
5 0.147050023078918
6 0.149100065231323
7 0.150699973106384
8 0
9 0
10 0
11 0
12 0
13 0
14 0
15 0
16 0
17 0
18 0
19 0
20 0
21 0
22 0
23 0
24 0
25 0
26 0
27 0
28 0
29 0
};
\addlegendentry{\small $\mathcal{R}_{8}$}
\addplot [thick, green!50.1960784313725!black, mark=diamond*, mark size=2, mark options={solid}]
table {%
0 0
1 0.151600003242493
2 0.0500999689102173
3 0.150449991226196
4 0
5 0
6 0.0883500576019287
7 0.140499949455261
8 0
9 0
10 0.140499949455261
11 0
12 0.142300009727478
13 0.0577000379562378
14 0.142300009727478
15 0.0577000379562378
16 0.142300009727478
17 0.0577000379562378
18 0.142300009727478
19 0.0577000379562378
20 0.142300009727478
21 0.0577000379562378
22 0.142300009727478
23 0.0577000379562378
24 0.142300009727478
25 0.0577000379562378
26 0.142300009727478
27 0.0577000379562378
28 0.142300009727478
29 0.0577000379562378
};
\addlegendentry{\small $\mathcal{R}_{12}$}
\addplot [thick, color0, mark=triangle*, mark size=2, mark options={solid}]
table {%
0 0
1 0.146700024604797
2 0.146899938583374
3 0.154250025749207
4 0
5 0
6 0.0521500110626221
7 0
8 0
9 0
10 0
11 0
12 0
13 0.142300009727478
14 0.0577000379562378
15 0.142300009727478
16 0.0577000379562378
17 0.142300009727478
18 0.0577000379562378
19 0.142300009727478
20 0.0577000379562378
21 0.142300009727478
22 0.0577000379562378
23 0.142300009727478
24 0.0577000379562378
25 0.142300009727478
26 0.0577000379562378
27 0.142300009727478
28 0.0577000379562378
29 0.142300009727478
};
\addlegendentry{\small $\mathcal{R}_{18}$}
\addplot [thick, color1, mark=x, mark size=2, mark options={solid}]
table {%
0 0
1 0
2 0
3 0
4 0.0497499704360962
5 0
6 0.151800036430359
7 0
8 0.00639998912811279
9 0
10 0
11 0
12 0
13 0
14 0
15 0
16 0
17 0
18 0
19 0
20 0
21 0
22 0
23 0
24 0
25 0
26 0
27 0
28 0
29 0
};
\addlegendentry{\small $\mathcal{R}_{26}$}
\end{axis}

\end{tikzpicture}
    \vspace{-0.20cm}
    \caption{Evolution of the density distribution of the swarm in the obstacle bins and the final bins (top figure), in the starting bins (middle figure), and in bins under the GTL constraint $\square (f_i \leq 0.15)$ (bottom figure). }
    \label{fig:density-evl-ex1}
\end{figure}
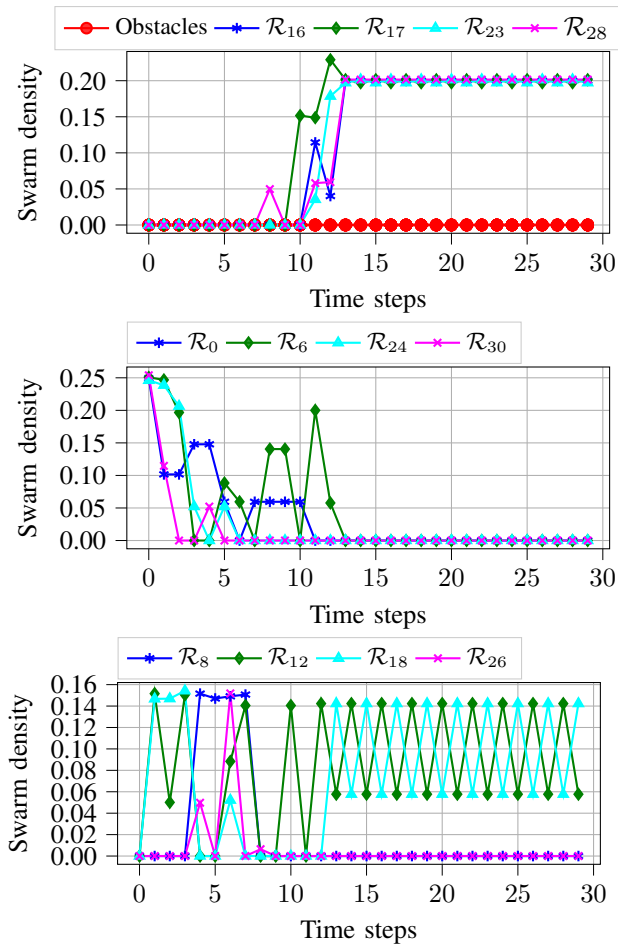

Figure~\ref{fig:density-evl-ex1} empirically demonstrates that the proposed approach is sound since all the specifications for this experiment were satisfied. 
Specifically, one can observe that the density inside the obstacles bins is always zero, all the capacity constraints are satisfied, and the final bins constraints (density greater than $0.2$) are also satisfied. 
Thus, we empirically demonstrate with this example the correctness of our algorithm. 
Furthermore, \controlalgo{} took only $1.1s$ to terminate.

\subsection{Heterogeneous Swarm in a Gridworld}
\begin{figure*}
    \centering
    \includegraphics[width=1.75in,height=1.75in]{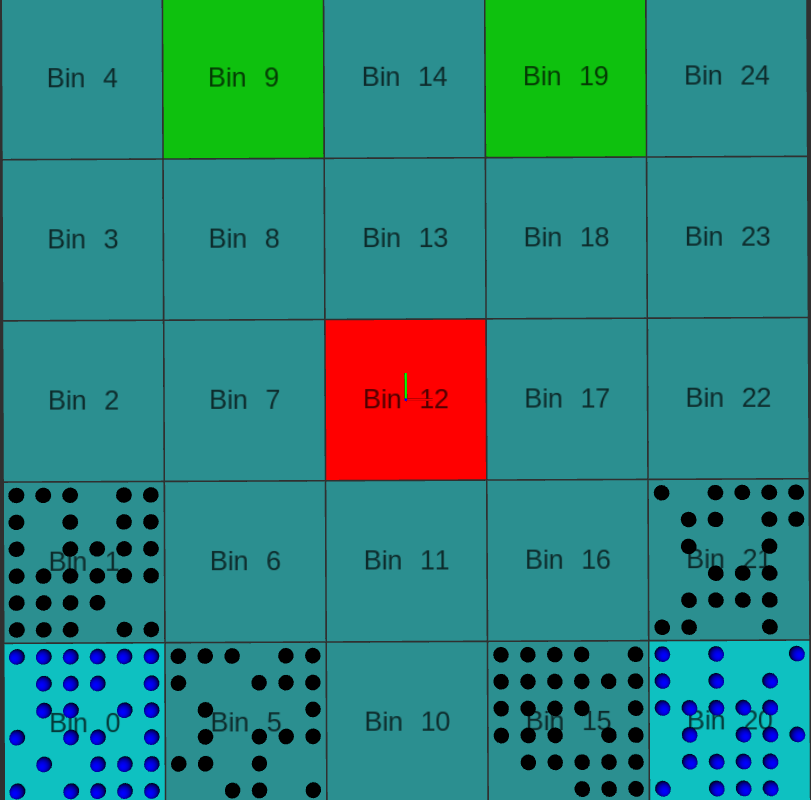}
    \includegraphics[width=1.75in,height=1.75in]{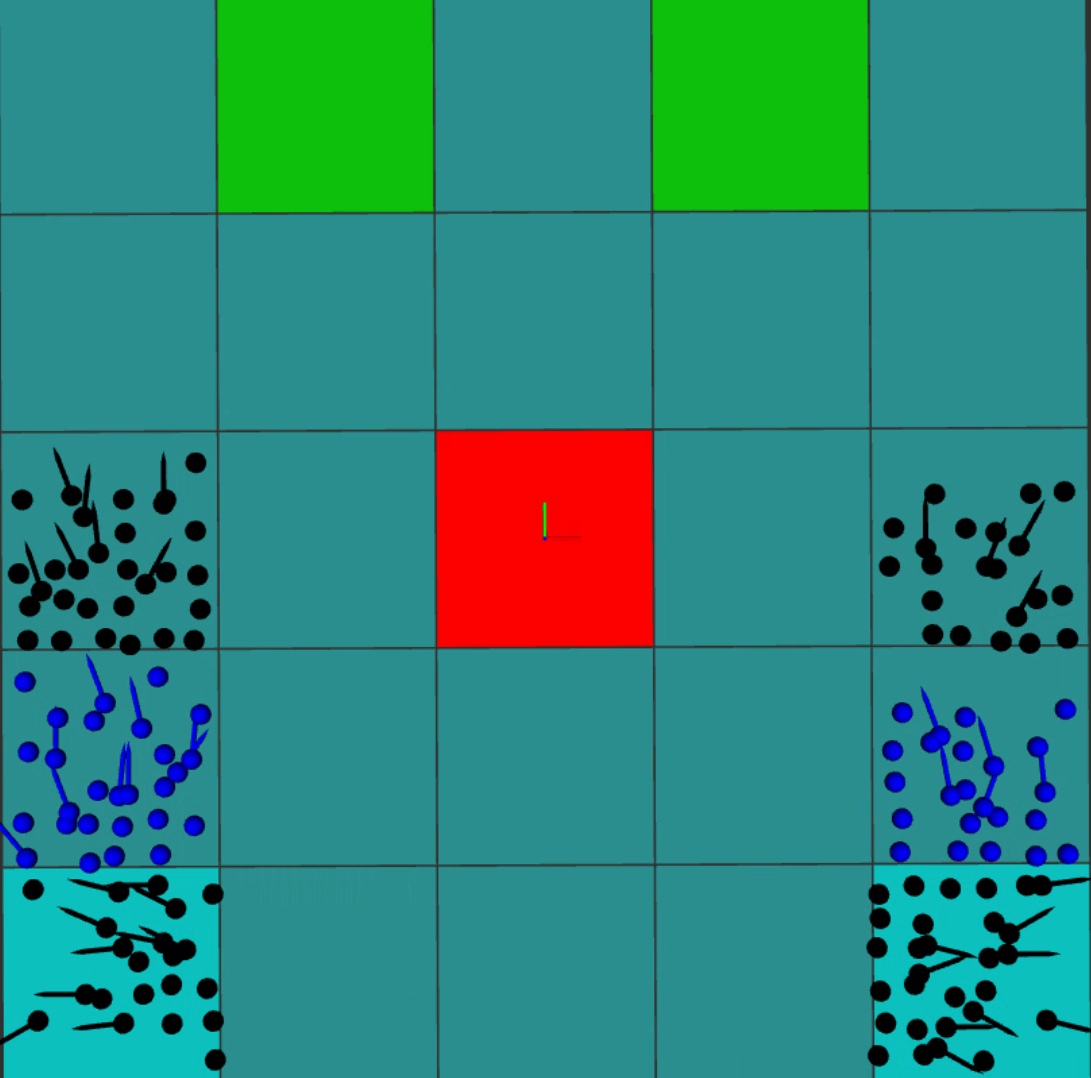}
    \includegraphics[width=1.75in,height=1.75in]{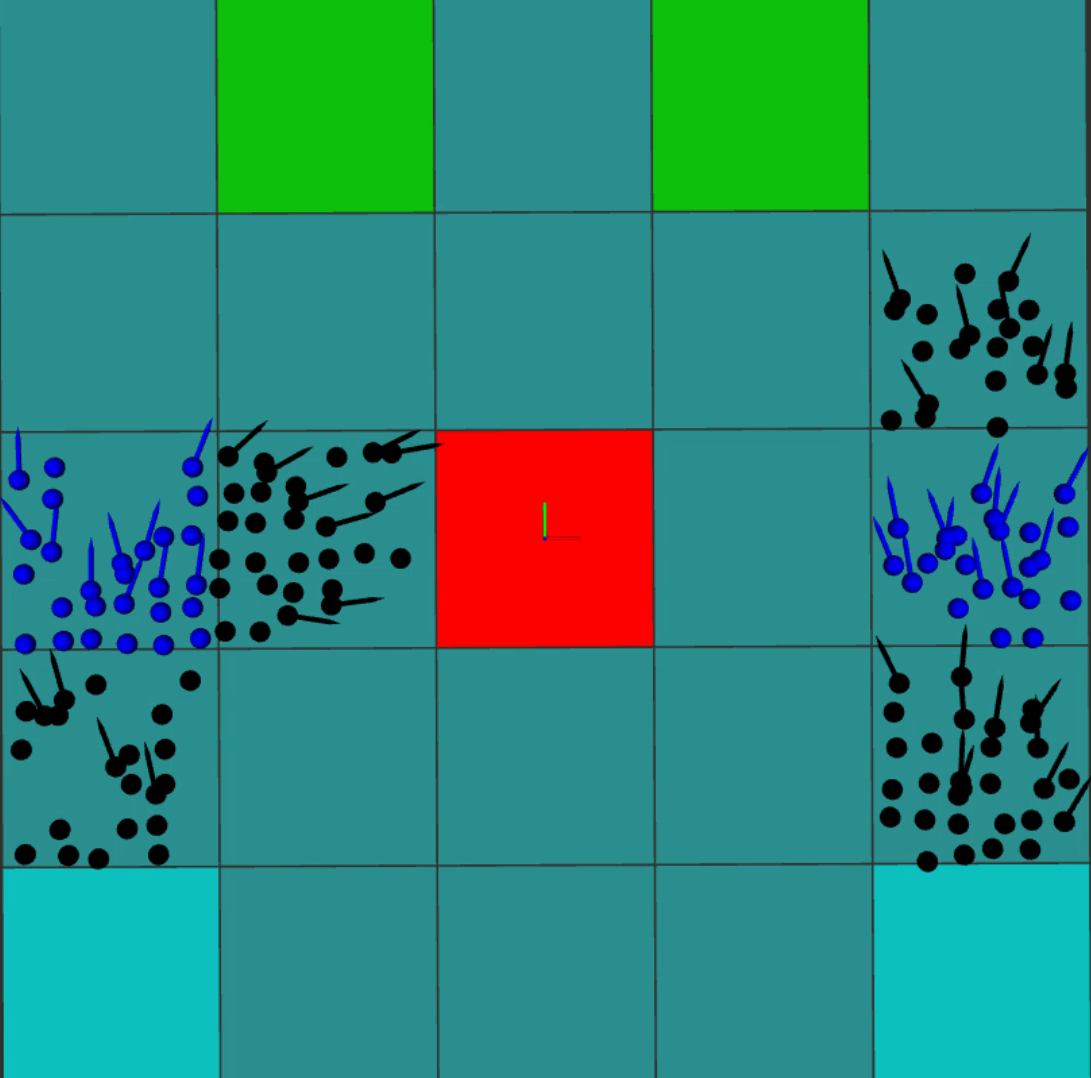}
    \includegraphics[width=1.75in,height=1.75in]{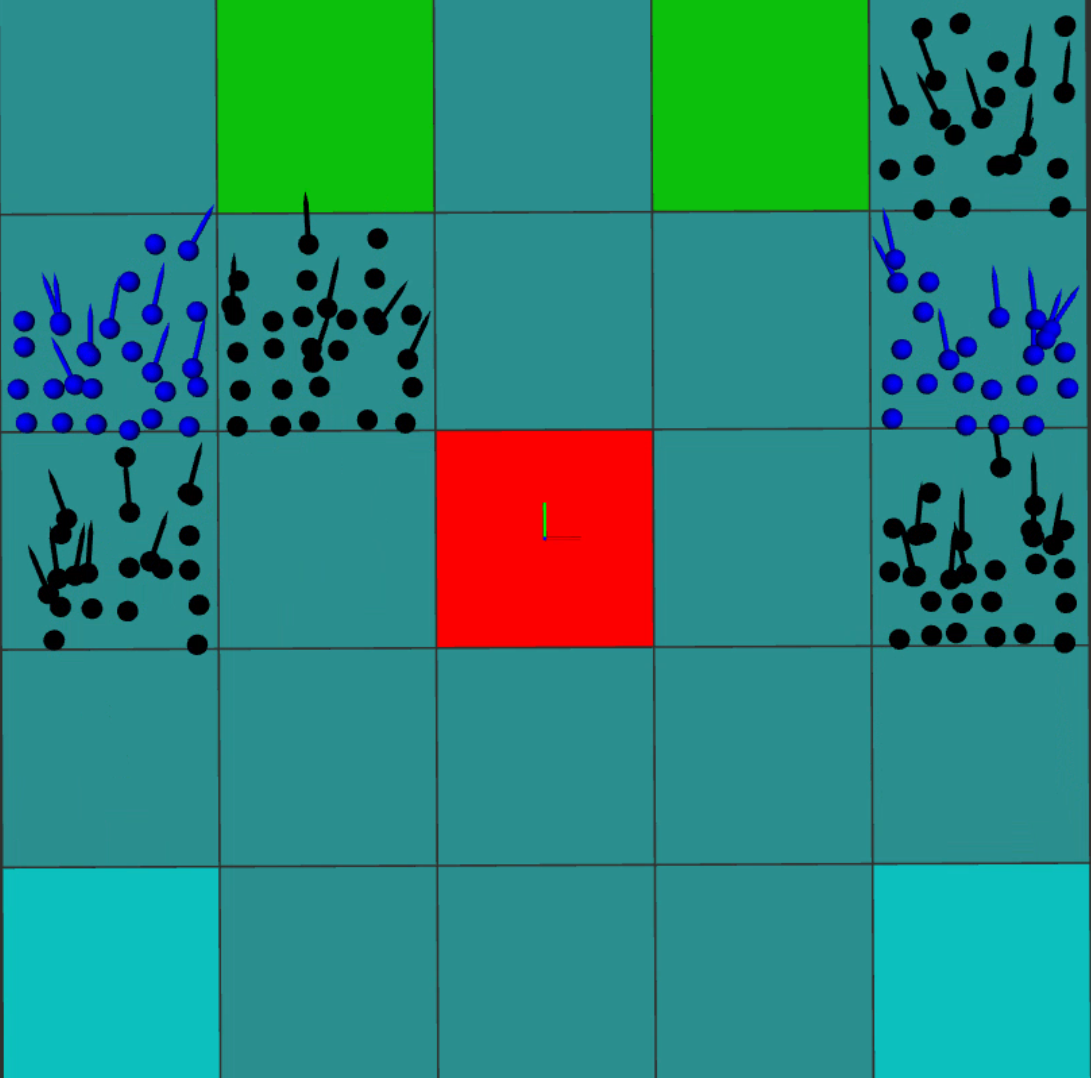}\vspace{0.10cm}
    \includegraphics[width=1.75in,height=1.75in]{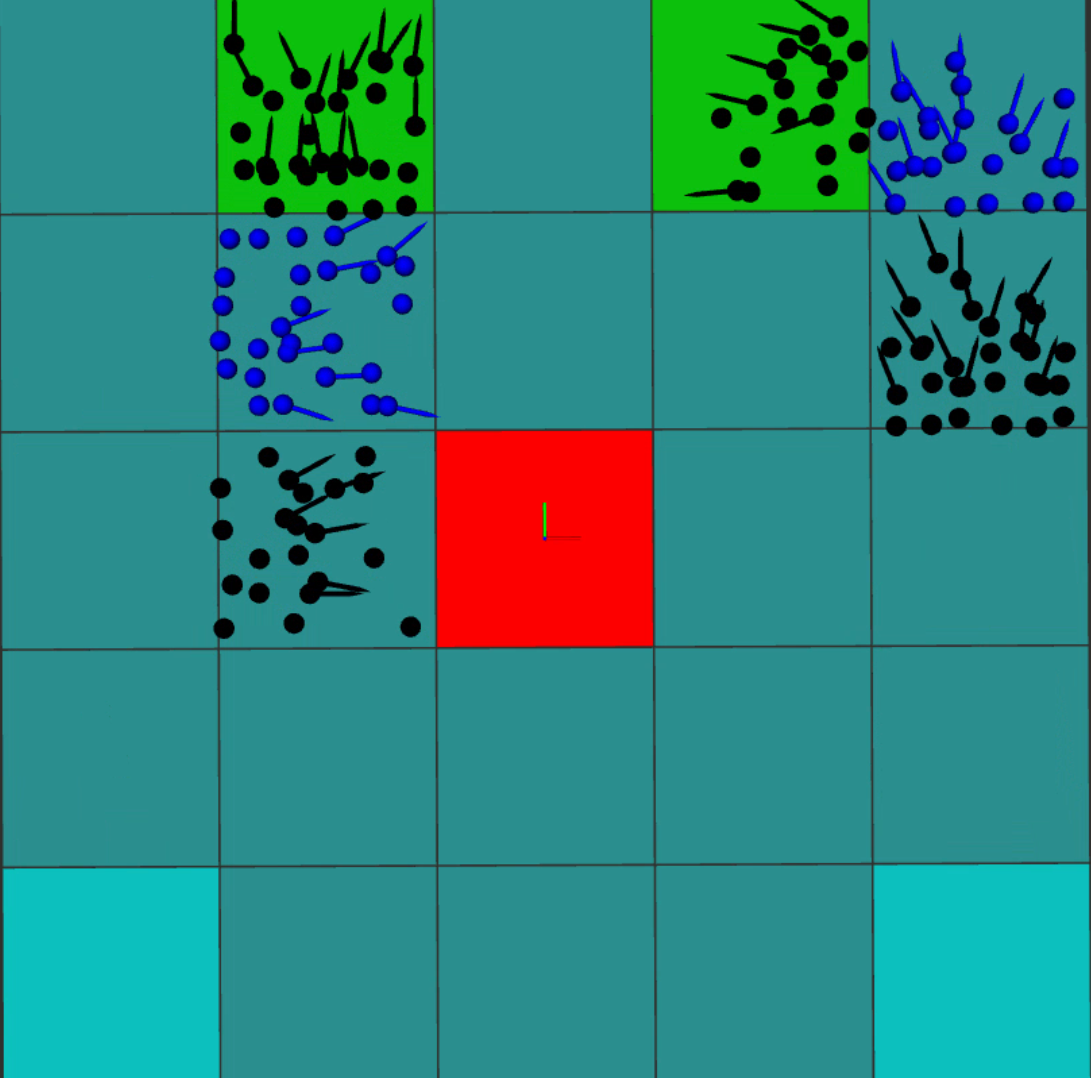}
    \includegraphics[width=1.75in,height=1.75in]{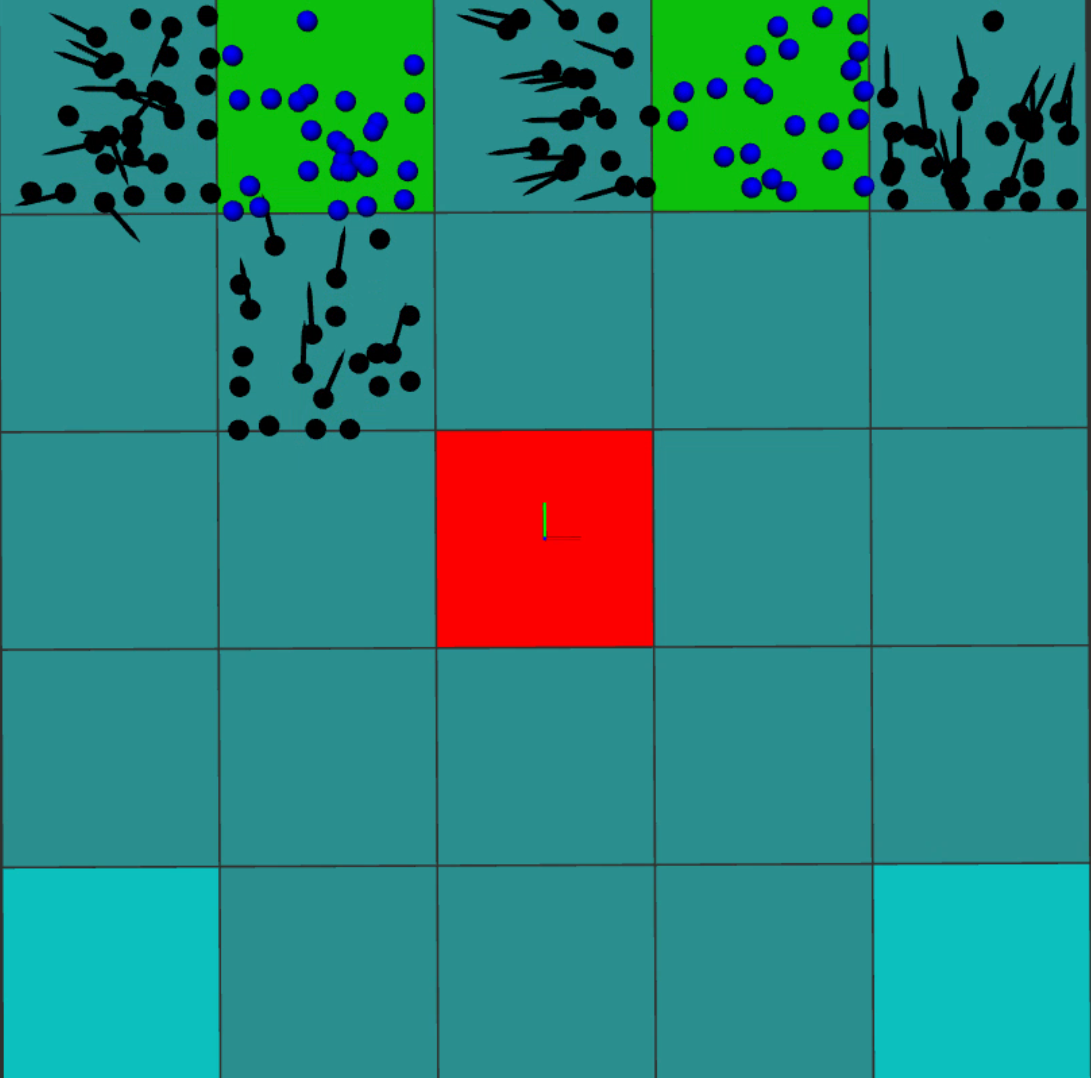}
    \vspace{-0.20cm}
    \caption{Evolution of the agents comprising the heterogeneous swarm at different time indexes. From the top left to the bottom right figure, we show the distribution of the agents at time indexes $0,1,\hdots,5$. In the figure, the obstacle bins are in red, the target and starting bins for the leaders are in green and cyan, respectively. Besides, we specify the followers agents as the black dots while the leaders are the blue dots.}\vspace{-0.40cm}
    \label{fig:heterogeneous-swarm}
\end{figure*}
In this section, we consider a simulation example with a swarm of heterogeneous agents navigating in a gridworld as shown in Figure~\ref{fig:heterogeneous-swarm}. The swarm contains $2$ sub-swarms.
The first sub-swarm, referred to as the leader, must reach a target density distribution while the second sub-swarm, referred to as the follower, must satisfy some density constraints in the neighborhood of each leader. 
Note that the leader and follower do not need to have the same dynamics. In such a situation, it is straightforward to deal with the different dynamics via the adjacency matrix $A_\mathrm{adj}^s$. However, for simplicity of this simulation, we assume that each sub-swarm has the same dynamics. 
Therefore, we obtain the graph representing the gridworld and label it with $f_{i}(x^1, x^2) = (x^1,x^2)$ for each node $i$. We use GTL to express the task specifications as follows:
\begin{itemize}
    \item Initially $x_i^0(0) = 0.5$ for $i\in\{0,20\}$ and we impose no constraints on the sub-swarm $1$. That is, the algorithm should find a correct initialization for the follower.
    \item We enforce the capacity constraints $\square (f_i \leq [0.5,0.25])$ for each bin $i$. Thus, each bin can contain at most $50\%$ of agents from sub-swarm $0$ and $25\%$ from sub-swarm $1$.
    \item We consider the GTL formulas $\Diamond \square (f_i \geq [0.5,0])$ for nodes $i \in \{9,19\}$ and $\square (f_i = [0,0])$ for the obstacle bin $i=12$. 
    Thus, the density of the leader sub-swarm $0$ should eventually be $0.5$ in bin $9$ and bin $19$.
    \item We require that no follower should be in the same bin as the leader with the formula $\square (f_i \leq [0,1] \vee f_i \leq [1,0])$ for all bins $i$.
    \item We enforce $\square (f_i \leq [0,1] \vee \exists^{2}\bigcirc (f_i \geq [0,0.25]))$ for all bins $i$. 
    This means that for each node, there should always be no leader in each bin $i$ or if there is a leader in any bin $i$, the swarm density of the follower sub-swarm should be greater than $0.25$ in at least two of the neighboring bins of $i$. 
    In other words, each leader should always be surrounded by followers.
\end{itemize}

In this scenario, the swarm is comprised of a total of $150$ agents performing collision avoidance in a decentralized manner. 
The leader swarm contains $50$ agents while the follower swarm contains $100$ agents.
Each agent in the simulation uses optimal reciprocal collision avoidance (ORCA)~\cite{jaimes2008approach} to dynamically and locally compute safe velocities to reach a given goal region. 
We first apply \controlalgo{} to find time-varying Markov matrices $M^1(t)$ and $M^2(t)$ such that the specifications above are satisfied. The computation time to generate the Markov matrices was $1.8$s.
Then, each agent independently and probabilistically chooses their target bin based on Algorithm~\ref{alg:psg-algo} with computed $M^s(t)$. 
When the target bin is obtained, the line~\ref{alg-local-interaction} of Algorithm~\ref{alg:psg-algo} consists of using ORCA to generate in real-time, at a fixed frequency, control velocities to reach the target bin while avoiding the fixed obstacles and the other agents in the gridworld.

Figure~\ref{fig:heterogeneous-swarm} demonstrates that the GTL specifications are satisfied. Specifically, it can be seen that the obstacles are always avoided, the leaders reach the target bins with the desired densities, the leader and follower never occupy the same bin at each time index, and finally the leaders are always surrounded by followers in at least two adjacent bins.
For example, the third image shows that the leader sub-swarm has $50\%$ of agents in bin $2$ and $50\%$ in bin $22$, which are surrounded by $25\%$ of follower agents in each bin $1$, $7$, $21$, and $23$.\looseness=-1



\section{Conclusion} \label{sec:conclusion}
We develop a correct-by-construction algorithm to control, in a decentralized and probabilistic manner, the density distribution of a swarm of heterogeneous agents subject to infinite-horizon GTL specifications. The algorithm, agnostic to the number of agents comprising the swarm, relies on synthesizing time-varying Markov matrices by adequately formulating the problem as either linear, semi-definite, or mixed-integer linear programs. The synthesized Markov matrices are independently used by each agent to determine the next targets while the entire swarm satisfies the specifications. Theoretically, we prove that the algorithm is correct by construction, and a complexity analysis shows that it significantly improves scalability over existing swarm control approaches. Empirically, we successfully demonstrated the efficiency and correctness of the algorithm in several simulation experiments.\vspace{-0.15cm}

\bibliographystyle{IEEEtran}
\bibliography{IEEEabrv,ref}\vspace*{-0.50cm}

\end{document}